\documentclass[a4paper,UKenglish,cleveref,autoref,thm-restate,numberwithinsect]{lipics-v2021}

\nolinenumbers

\usepackage{amsmath,amssymb,amsfonts}
\usepackage{xspace}
\usepackage{xcolor}
\usepackage{textcomp} 
\usepackage[most]{tcolorbox}

\newcommand{\ia}{\textit{i}}
\newcommand{\ib}{\textit{ii}}
\newcommand{\ic}{\textit{iii}}
\newcommand{\iiv}{\textit{iv}}

\newcommand{\calN}{\mathbb{N}}
\newcommand{\calA}{\mathcal{A}}

\newcommand{\calS}{\mathcal{S}}

\newcommand{\calB}{\mathcal{B}}

\newcommand{\calF}{\mathcal{F}}

\newcommand{\andy}[1]{\textcolor{red}{}}
\newcommand{\udi}[1]{\textcolor{blue}{}}

\newcommand{\temph}[1]{\textbf{#1}}

\newcommand{\mypara}[1]{\smallskip\noindent\textbf{#1.}}

\title{Volition-Guarded Multiagent Atomic Transactions: Describing People and their Machines}
\titlerunning{Describing People and their Machines}

\author{Andy Lewis-Pye$^1$ \& Ehud Shapiro$^{1,2}$}{$^1$London School of Economics, UK, and $^2$Weizmann Institute of Science, Israel}{{a.lewis7,e.shapiro2}@lse.ac.uk}{}{}

\authorrunning{Lewis-Pye and Shapiro}

\Copyright{Andy Lewis-Pye and Ehud Shapiro}

\hideLIPIcs
\ArticleNo{0}



\begin{document}

\maketitle

\begin{abstract}
Formal models for concurrent and distributed systems describe machines; the people who operate them are either ignored or treated as external environment. Yet, key distributed systems---notably grassroots platforms---include people operating their personal machines (smartphones), and their faithful description must include the states of both people and machines and how they jointly effect system behaviour. 

Here, we propose volition-guarded multiagent atomic transactions---executed atomically by machines and guarded by their people's volitions---as a novel mathematical foundation for specifying systems consisting of people operating machines. 
Each agent's state consists of a volitional state and machine state; a transaction is enabled when the machine precondition holds and the guarding persons are willing.
For example, befriending two people is guarded by both; unfriending, by either; voluntary swap of coins and bonds is guarded by both parties, while a payment is guarded by the payer.

We develop the mathematical machinery to express safety and liveness of platforms specified in this framework, to implement one platform by another, and for an implementation to be resilient to faults; and provide example specifications of two grassroots platforms: social networks, and coins and bonds.  These specifications are then used by AI to derive working implementations.
We employ here a novel and simpler definition of `grassroots' that better captures the informal notion---multiple instances can form and operate independently, yet may coalesce---and show that the platforms specified here are grassroots under the new definition.  We further introduce \emph{volitionally grassroots} protocols, in which two groups can become connected only by mutual consent---the first transaction coupling them must be willed by a member of each---and show that both platforms are volitionally grassroots.
\end{abstract}

\section{Introduction}

\mypara{The gap} Formal models for concurrent and distributed systems describe machines; the people who operate them are either ignored or treated as external environment (see Section~\ref{section:related-work}).  Across all surveyed formal traditions---Turing's choice machines~\cite{turing1936computable}, CSP~\cite{hoare1985communicating}, CCS~\cite{milner1980calculus}, I/O automata~\cite{lynch1989introduction}, angelic/demonic nondeterminism~\cite{back1998refinement}, ATL~\cite{alur2002alternating}, game semantics~\cite{abramsky2000full}, ceremony analysis~\cite{ellison2007ceremony}, electronic institutions~\cite{esteva2001formal}---the person is modelled as environment, opponent, error source, or unconstrained nondeterministic process, but always as an entity external to the agent, never as a formal component of the agent's state.

\mypara{Proposal}
Here, we propose \emph{volition-guarded multiagent atomic transactions}---executed atomically by machines and guarded by their people's volitions---as a formal foundation for specifying systems consisting of people operating machines.  Each agent's state decomposes into a \emph{machine state} and a \emph{volitional state}; a transaction is enabled when the machine precondition holds and the guarding persons are willing.  Volitions are thus persistent, inspectable agent state rather than point-of-choice nondeterminism, environment, or opponent moves as in prior formal traditions.

\mypara{Example: Grassroots social graphs}  Each agent $p$ maintains a set of friends $c_p \subseteq P$.  The social graph evolves via befriending and unfriending, specified as volition-guarded transactions:
\begin{tcolorbox}[colback=gray!5!white,colframe=black!75!black,top=2pt,bottom=2pt]
\begin{enumerate}
    \item \textbf{Befriend}: $c'_p := c_p \cup \{q\}$, $c'_q := c_q \cup \{p\}$, provided $q \notin c_p$.  Guarded by $\{p, q\}$.
    \item \textbf{Unfriend}: $c'_p := c_p \setminus \{q\}$, $c'_q := c_q \setminus \{p\}$, provided $q \in c_p$.  Guarded by $p$ or $q$.
\end{enumerate}
\end{tcolorbox}
\noindent Befriending requires both persons to be willing; unfriending can be initiated by either.  The guard determines which persons must be willing for a machine transaction to proceed.  This distinction is the essence of volition-guarded multiagent atomic transactions.  Further examples---including child-safe social networking and coins and bonds---appear in Section~\ref{subsection:examples}.

\mypara{Motivating domain: grassroots platforms} Volition-guarded multiagent atomic transactions are a general framework for systems of people operating machines.  An example of a class of such systems, which motivated this work, is grassroots platforms~\cite{shapiro2023grassrootsBA,shapiro2023gsn,shapiro2024gc,shapiro2025GF}, which aim to offer an egalitarian alternative to global platforms, centralized and decentralized alike.   Grassroots platforms consist of people operating their personal machines (smartphones), and can have multiple instances that emerge and operate independently of each other and of any global resource except the network, yet can interoperate and coalesce once interconnected.  Key grassroots platforms include grassroots social networks~\cite{shapiro2023gsn,shapiro2026volitional}, grassroots coins~\cite{shapiro2024gc,lewis2023grassroots} and bonds~\cite{shapiro2026bonds}, and grassroots democratic communities~\cite{halpern2024federated,shapiro2025GF,keidar2025constitutional}.  No platform that operates on a shared global resource---a replicated ledger (Blockchain~\cite{bitcoin}), a distributed data structure (IPFS~\cite{benet2014ipfs}, DHT~\cite{maymounkov2002kademlia}), or a distributed pub/sub system with a global directory~\cite{chockler2007constructing,chockler2007spidercast,buchegger2009peerson})---is grassroots.  Mastodon~\cite{raman2019challenges} is peer-to-peer among servers not people; BitTorrent~\cite{torrentfreak2021bittorrent} requires reaching a tracker or DHT boot nodes to join a swarm; Scuttlebutt~\cite{tarr2019secure,kermarrec2020gossiping} is grassroots in design even if not formally proven as such.  

An earlier paper~\cite{shapiro2025atomic} introduced multiagent atomic transactions for grassroots platforms and presented transactions-based specifications of social graphs, cryptocurrencies, and democratic federations; that work did not address the role of people, transaction equivalence, or liveness, and employed the original definition of grassroots protocols. 

\mypara{Framework overview} We develop the mathematical machinery needed to express volition-guarded multiagent atomic transactions and the safety and liveness of grassroots platforms they specify, to implement one platform by another, and for an implementation to be resilient to faults.  We introduce \emph{transaction equivalence classes}, which provide a natural notion of ``same action'' across configurations, and define liveness in terms of such classes: a run is correct iff every enabled class is eventually taken.  We demonstrate the framework with two grassroots platforms:
\begin{enumerate}
\item \textbf{Grassroots social networks}~\cite{shapiro2023gsn,shapiro2026volitional} let people maintain their friendship connections through local storage and peer-to-peer relationships without central control, with the social graph evolving through befriending and unfriending.
\item \textbf{Grassroots coins and bonds}~\cite{shapiro2024gc,lewis2023grassroots,shapiro2026bonds} let each person mint their own coins, backed by the goods and services they offer, and exchange them with others via atomic swaps; bonds extend coins with a maturity date, enabling interest-bearing credit, loans, and the full gamut of financial instruments.
\end{enumerate}
The specifications presented here have been used by AI to derive working implementations in GLP, a grassroots multiagent concurrent logic programming language~\cite{shapiro2025glp}, as reported in companion papers~\cite{shapiro2026implementing,shapiro2026volitional,shapiro2026bonds,shapiro2026types}.  We employ here a novel---simpler than the original~\cite{shapiro2023grassrootsBA,shapiro2025atomic}---definition of grassroots protocols based on interleavings of correct runs, that better captures the informal notion of grassroots, and show that the platforms specified here are grassroots under the new definition.  We strengthen this to \emph{volitionally grassroots} protocols, in which two groups can become connected only by the mutual consent of a member of each---the first transaction coupling them is guarded on both sides---and prove both platforms volitionally grassroots.  We also prove that Bitcoin, distributed hash tables, and similar systems are not grassroots.

\mypara{Paper outline} 
Section~\ref{section:dts} presents examples of grassroots platforms specified by volition-guarded transactions, and introduces the formal framework: volition-guarded multiagent atomic transactions and the volitional transactions they induce, transaction equivalence, and liveness.
Section~\ref{section:gs-protocols} defines protocols and grassroots protocols via interleavings of correct runs, gives a syntactic guard condition under which a transactions-based protocol is grassroots, and defines volitionally grassroots protocols, in which every first interaction between two groups is willed by a member of each.
Section~\ref{section:platforms} presents two grassroots platforms (social networks, coins and bonds), proves their safety properties, and proves that they are volitionally grassroots.
Section~\ref{section:related-work} discusses related work.
Section~\ref{section:conclusion} concludes and discusses future work.
The appendix contains proofs that Bitcoin, distributed hash tables, and IPFS are not grassroots (Appendix~\ref{appendix:not-grassroots}), notes on implementation (Section~\ref{section:implementation}), and a detailed survey of formal models of persons in concurrent systems (Appendix~\ref{appendix:persons}).

\section{Volitional Multiagent Transition Systems}\label{section:dts}

Earlier work introduced multiagent transition systems~\cite{shapiro2021multiagent}, grassroots protocols and platforms~\cite{shapiro2023grassrootsBA}, and their definition via multiagent atomic transactions~\cite{shapiro2025atomic}---capturing the behaviour of machines but not of the people operating them.  Here we develop the mathematical machinery to describe agents consisting of a person and a machine, and the volitional multiagent atomic transactions they execute.

A volitional transaction is a ``regular'' multiagent atomic transaction---henceforth, a \emph{machine transaction}---guarded by the volitions of some, all, or none of the people whose machines participate; a person's volitional state is a set of equivalence classes of machine transactions they are willing their machine to participate in.  The social graph illustrates the two extremes: befriending is guarded by both $p$ and $q$ (the class of the `befriend $p$ and $q$' machine transaction must be in both volitions), while unfriending is guarded by either.  A person may freely change their volitional state via change-volition transactions; additionally,  an equivalence class is removed from every agent's volitional state when a machine transaction in this class is taken, fulfilling the will.

The definitions of volitional transactions and agent states are mutually-recursive: a volitional transaction is a machine transaction guarded by volitions, which are themselves sets of equivalence classes of machine transactions.
We resolve this circularity bottom-up: first machine transactions, then their equivalence, then agent states (including volitional states), and finally volitional transactions.

\subsection{Agents, People, Machines, Volition-Guarded and Volitional Transactions}\label{subsection:at}

We assume a potentially infinite set of \emph{agents} $\Pi$, but consider only finite subsets of it, so when referring to a particular set of agents $P \subset \Pi$ we assume $P$ to be nonempty and finite.  We use $\subset$ to denote the strict subset relation and $\subseteq$ when equality is also possible, and use $p\ne q \in P$ as a shorthand for $p\in P \wedge q\in P \wedge p\ne q$.
As standard, we use $S^P$ to denote the set of all total functions from $P$ to $S$, and if $c\in S^P$ we use $c_p$ (instead of $c(p)$) to denote the value of $c$ at $p\in P$.

\begin{definition}[Machine State, Configuration, Transaction, Volition-Guarded Multiagent Atomic Transaction]\label{definition:mt} 
Given an arbitrary set $S$ of \temph{machine states}, with a designated \temph{initial state} $s_0 \in S$, and agents $Q \subset \Pi$, a \temph{machine configuration} over $Q$ is a member of $S^Q$, and a \temph{machine transaction} over \temph{participants} $Q$ is a pair $c\rightarrow c' \in (S^Q)^2$ such that $c\ne c'$. Given such a machine transaction $t$, a 
\temph{volition-guarded multiagent atomic transaction} over $t$---henceforth, \temph{volition-guarded transaction}---is a pair $(t,Q')$ where $Q'\subseteq Q$ are its \temph{guards}.
\end{definition}

Machine transactions are atomic and asynchronous~\cite{shapiro2021multiagent}---they can be carried out by their participants at any time, regardless of the states of non-participants.  Participants include both active agents (whose state changes) and stationary agents (whose state is a precondition but does not change).  Volition-guarded transactions are machine transactions that can be carried-out only if their guards $Q'\subseteq Q$ are willing.  Volition-guarded transactions do not distinguish between agents that initiate a transaction and those willing to participate in it.  When we say a transaction is ``guarded by $\{p,q\}$,'' both must be willing; when we say it is ``guarded by either $p$ or $q$,'' we mean there are two volition-guarded transactions over the same machine transaction, $(t,\{p\})$ and $(t,\{q\})$, so that either person's volition suffices.

Distinct machine transactions can represent ``the same action'' in different configurations; we capture this with an equivalence relation on machine transactions.
\begin{definition}[Transaction Equivalence]\label{definition:equivalence}
Given a set of machine transactions $R$, a \temph{transaction equivalence} is an equivalence relation $\sim$ on $R$ such that $t \sim t'$ implies $t$ and $t'$ have the same participants.  We write $[t]$ for the equivalence class of $t$ under $\sim$.
\end{definition}
For example, all befriend$(p,q)$ transactions---differing only in the configurations in which they occur---form an equivalence class.   Further examples for each platform appear in Section~\ref{section:platforms}.

Next we provide the mathematical machinery allowing people to express their volitions regarding the classes of machine transactions their machines may participate in.

\begin{definition}[Agent State and Configuration]\label{definition:agent-state}
Given agents $P$, states $S$ with initial state $s_0$, a set of machine transactions $T$ each over its own participants $Q\subseteq P$ and $S$, and equivalence $\sim$ on $T$, an \temph{agent state} is a pair $(V,m)\in \calA = (2^{T/\sim}  \times S)$ where  $V$ is its \temph{volitional state} and $m \in S$ its  \temph{machine state}.  The \temph{initial agent state} is $(\emptyset,s_0)$.  An \temph{agent configuration} $c$ over $P$, $S$, $T$, and $\sim$ is a member $c\in \calA^P$ in which $c^v_p \subseteq (T/{\sim})_p$ for every $p\in P$, where $(T/{\sim})_p$ denotes the classes in $T/{\sim}$ in which $p$ is a participant; we write $c^v_p$  for the volitional state and $c^m_p$ for the machine state of agent $p$ in $c$.
\end{definition}

\begin{definition}[Volitional Multiagent Atomic Transaction]\label{definition:vmat}
Given agents $P$, states $S$, machine transactions $T$ over $P$ and $S$, and equivalence $\sim$ on $T$:
\begin{enumerate}
    \item A \temph{change-volition transaction of agent $p\in P$} is a pair $c\rightarrow c'$ of agent configurations over $\{p\}$, $S$, $T$, and $\sim$ such that $c^v_p, c'{^v_p} \subseteq (T/{\sim})_p$ and $c^v_p \ne c'{^v_p}$, and $c^m_p = c'{^m_p}$.
    \item A \temph{volitional machine transaction} induced by a volition-guarded machine transaction $(t,Q')$, for some $t= (d\rightarrow d')\in T$ over $Q'\subseteq Q\subseteq P$, is a pair $c\rightarrow c'$ where $c\ne c'$ are agent configurations over $P$, $S$, $T$, and $\sim$ such that $[t]\in c^v_q$ for every $q\in Q'$; $c^m_p = d_p$ and $c'{^m_p} = d'_p$ for every $p\in Q$; $c^m_p = c'{^m_p}$ for every $p\in P\setminus Q$; and $c'{^v_p} = c^v_p \setminus \{[t]\}$ for every $p\in P$.
    \item A \temph{volitional multiagent atomic transaction} is a change-volition transaction or a volitional machine transaction. 
\end{enumerate}
\end{definition}
When a volitional machine transaction induced by $(t,Q')$ is taken, the class $[t]$ is removed from every agent's volitional state: the will is fulfilled by any equivalent transaction.  A person may independently change their volitional state via change-volition transactions, which may add or remove classes; beyond these, the framework removes a class from $c^v_p$ only upon fulfilment. 

\mypara{Sybils} Sybils~\cite{douceur2002sybil} arise when a person operates two or more identities. Defending against sybils is a broad research area~\cite{levine2006survey,alvisi2013sok}, which is largely algorithmic, often evaluated experimentally. A formal treatment of sybils requires a formal notion of a person and of the correspondence between persons and the identifiers they own~\cite{shahaf2020genuine}. The present framework offers a natural vehicle for formally studying sybils and sybil resilience: since an agent is a person operating a machine, identified by the public key of the keypair its machine holds, a sybil is a single person operating two or more agents under distinct keypairs, and absent sybils agents are in one-to-one correspondence with persons. Using the current framework to formally study sybil resilience is a subject of future work.

\subsection{Examples of Grassroots Platforms Specified by Volition-Guarded Multiagent Atomic Transactions}\label{subsection:examples}

A grassroots platform is specified by a set of volition-guarded transactions over a local-states function; the formal machinery for deriving a transition system and proving it grassroots is introduced in this section.  First, we present the volition-guarded transactions that specify several grassroots platforms, illustrating the range of guard structures that arise in practice.

\mypara{Child-safe social networks~\cite{shapiro2026volitional}}  A child-safe social network is a grassroots social network in which a child's partaking in online activities is subject to parental consent. Each agent $a$ maintains a set of friends $c_a\subseteq P$.  In what follows, $r$, $s$ are children with respective parents $p$, $q$ (where $r,s,p,q$ are four distinct agents with the stated precondition).
\begin{tcolorbox}[colback=gray!5!white,colframe=black!75!black,top=2pt,bottom=2pt]
\begin{enumerate}
    \item \textbf{Child befriend}: $c'_r := c_r \cup \{s\}$, $c'_s := c_s \cup \{r\}$, provided $q \in c_p$ and $s \notin c_r$.  Guarded by $\{r, s, p, q\}$.
    \item \textbf{Child unfriend}: $c'_r := c_r \setminus \{s\}$, $c'_s := c_s \setminus \{r\}$, provided $s \in c_r$.  Guarded by any one of $\{r, s, p, q\}$.
\end{enumerate}
\end{tcolorbox}
\noindent The precondition $q\in c_p$ requires the parents to be friends.  Child befriending requires all four---both children and both parents---to be willing; child unfriending can be initiated by any one of the four.  The parent--child assignment is fixed externally and not part of the agent state; the full formalisation appears in~\cite{shapiro2026volitional}. 

\mypara{Grassroots coins~\cite{shapiro2024gc}}  Grassroots coins are units of personal debt: each person mints their own coins, backed by the goods and services they offer, and liquidity arises from mutual credit via coin exchange among persons that know and trust each other.  Each agent $p$ maintains a multiset of coins $c_p$; we write $\text{\textcent}_r$ for a coin minted by $r$ (a \emph{$r$-coin}) and $\text{\textcent}_r^k$ for a multiset of $k$ such coins.  We write $\cup$ and $\setminus$ for multiset union and difference throughout.
\begin{tcolorbox}[colback=gray!5!white,colframe=black!75!black,top=2pt,bottom=2pt]
\begin{enumerate}
    \item \textbf{Mint}: $c'_p := c_p \cup \text{¢}_p^k$, $k>0$.  Guarded by $p$.
    \item \textbf{Voluntary swap}: $c'_p:= (c_p\cup  y) \setminus  x$, $c'_q:= (c_q\cup  x) \setminus  y$, provided $x \subseteq c_p$, $y \subseteq c_q$.  Guarded by $\{p,q\}$.
    \item \textbf{Pay}: $c'_p:= c_p \setminus  x$, $c'_q:= c_q\cup  x$, where $x$ is a set of $q$-coins, $x \subseteq c_p$.  Guarded by $p$.
    \item \textbf{Redeem}: $c'_p:= (c_p\cup  y) \setminus  x$, $c'_q:= (c_q\cup  x) \setminus  y$, where $x = \text{¢}_q^k \subseteq c_p$, $y \subseteq c_q$, $|y|=k$.  Guarded by $p$.
\end{enumerate}
\end{tcolorbox}
\noindent Minting is a personal decision; voluntary swaps require both parties to be willing; payments and redemptions are guarded by the payer/redeemer.  In redemption, the redeemer chooses any $k$ coins held by the issuer.  Grassroots bonds~\cite{shapiro2026bonds} extend grassroots coins with maturity dates, enabling interest-bearing credit and loans; they are formally specified in Section~\ref{section:platforms}.

\mypara{Summary}  The guard captures the essential distinction between voluntary and obligatory transactions.  Befriending and voluntary swaps are guarded by all participants---they require mutual willingness.  Unfriending, payments, and redemptions are guarded by a single party---they are obligatory once initiated.  Child befriending illustrates the most complex case: a quaternary guard requiring all four persons to be willing.  Unguarded transactions (guard $Q'=\emptyset$), such as group message delivery~\cite{shapiro2026volitional}, are purely mechanical and require no volitions.

\subsection{Volitional Multiagent Transition Systems}\label{subsection:mts}

The following is a simplified variation, sufficient for the purpose of this work, on the foundations introduced in~\cite{shapiro2021multiagent}.   

\begin{definition}[Transition System, Computation, Run, Safe, Live, Correct]\label{definition:ts}\label{definition:liveness}
A \temph{transition system} is a tuple $TS=(S,s_0,T,{\sim})$, where: 
\begin{enumerate}
    \item $S$ is an arbitrary non-empty set, referred to as the set of \temph{states}.
    \item Some $s_0\in S$ is the designated \temph{initial state}. 
    \item $T\subseteq S^2$ is a set of \temph{correct transitions over} $S$, where each transition $t\in T$ is a pair $(s,s')$ of non-identical states $s\ne s'\in S$, also written as $t=s\rightarrow s'$.
    \item $\sim$ is a \temph{partial equivalence relation} on $T$: a symmetric and transitive relation on $T$, not necessarily reflexive.  Its domain $\{t\in T : t\sim t\}$ is partitioned into \temph{liveness classes} $T/{\sim}$; a transition outside the domain belongs to no class.
\end{enumerate} 
A \temph{computation} of $TS$ is a (nonempty, potentially infinite) sequence of states $r= s_1,s_2,\cdots$; it is a \temph{run} of $TS$ if $s_1=s_0$.
A computation $r= s_1,s_2,\ldots$ is \temph{safe}, also written $r\subseteq T$, if $s_i\rightarrow s_{i+1}\in T$ for every two consecutive states; and $s\xrightarrow{*}s'\subseteq T$ denotes the existence of a safe computation from $s$ to $s'$ (empty if $s=s'$).
A class $[t]\in T/{\sim}$ is \temph{enabled} in a state $s$ if $s\rightarrow s'\in[t]$ for some $s'\in S$.
A run $r$ is \temph{live} if no class $[t]\in T/{\sim}$ is enabled in every state of some suffix of $r$ with no member of $[t]$ occurring in the suffix.  A run is \temph{correct} if it is safe and live.
\end{definition}
A partial equivalence, rather than a total one, is used so that some transitions may carry no liveness requirement at all: only transitions in the domain of $\sim$ form classes and thereby incur a liveness obligation, while transitions outside the domain---belonging to no class---may occur in a correct run but are never required to.

\begin{definition}[Multiagent Transition System]\label{definition:dts-cd}
Given agents $P \subset \Pi$ and an arbitrary set $S$ of \temph{states} with a designated \temph{initial state} $s_0\in S$, 
a \temph{multiagent transition system} over $P$ and $S$ is a transition system $TS= (C,c_0,T,{\sim})$ with \temph{configurations} $C:= S^P$, \temph{initial configuration}  $c_0:= \{s_0\}^P$, \temph{transitions} $T\subseteq C^2$ a set of transactions over $P$ and $S$, and $\sim$ a partial equivalence on $T$. 
\end{definition}
Unary multiagent transition systems were introduced in~\cite{shapiro2021multiagent} and were employed to define the notion of grassroots protocols~\cite{shapiro2023grassrootsBA} and to provide unary specifications for various grassroots platforms~\cite{shapiro2023gsn,shapiro2024gc,lewis2023grassroots}.
Here, we employ $k$-ary transition systems, for any $k\le |P|$, in which several agents can change their state simultaneously.

Rather than specifying a multiagent transition system over a set of agents $P$ directly, we specify it via 
machine transactions (Definition~\ref{definition:mt}).

A machine transaction over $Q\subseteq P$ defines a set of multiagent transitions over $P$ in which all members of $P\setminus Q$ are stationary:
\begin{definition}[Transaction Closure]\label{definition:closure}
Let $P\subset \Pi$, $S$ a set of machine states, and $C:=S^P$.
For any transition or transaction $t = c\to c'$, we write $t_q := c_q\to c'_q$ and say $p$ is \temph{stationary} in $t$ if $c_p = c'_p$.
For a machine transaction $t=(c\rightarrow c')$ over $S$ with participants $Q$, the \temph{$P$-closure of $t$}, $t{\uparrow}P$, is the set of transitions over $P$ and $S$ defined by:
$$
t{\uparrow}P := \begin{cases} \{ t' \in C^2  :
\forall q\in Q.(t_q = t'_q) \wedge \forall p\in P\setminus Q.(p\text{ is stationary in }t')\} & \text{if } Q\subseteq P \\
\emptyset & \text{otherwise}
\end{cases}
$$
If $R$ is a set of machine transactions, each $t\in R$ over some $Q$ and $S$, then the 
\temph{$P$-closure of $R$}, $R{\uparrow}P$, is the set of transitions over $P$ and $S$ defined by:
$$
R{\uparrow}P := \bigcup_{t\in R} t{\uparrow}P
$$
Given a relation ${\sim}$ on $R$, its \temph{$P$-closure} ${\sim}{\uparrow}P$ is the relation on $R{\uparrow}P$ with $\hat t \mathrel{({\sim}{\uparrow}P)} \hat t'$ iff $\hat t \in t{\uparrow}P$ and $\hat t' \in t'{\uparrow}P$ for some $t \sim t'$.
\end{definition}
Namely, the closure over $P\supseteq Q$ of a machine transaction $t$ over $Q$ includes all transitions $t'$ over $P$ in which members of $Q$ do the same in $t$ and in $t'$, and the rest remain in their current (arbitrary) state.  The closure likewise carries any relation on transactions to one on the induced transitions.  If distinct transactions in $R$ have disjoint $P$-closures---as when every participant of every transaction in $R$ changes state---then each transition in $R{\uparrow}P$ has a unique inducing transaction, ${\sim}{\uparrow}P$ relates two transitions exactly when their inducing transactions are $\sim$-related, and ${\sim}{\uparrow}P$ is a partial equivalence whenever ${\sim}$ is.

\begin{lemma}[Compositionality of Closure]\label{lemma:closure-compositional}
For sets of machine transactions $R, R'$ over $S$ and any $P\subset \Pi$: $(R\cup R'){\uparrow}P = R{\uparrow}P \cup R'{\uparrow}P$, and $R\subseteq R'$ implies $R{\uparrow}P \subseteq R'{\uparrow}P$.
\end{lemma}
\begin{proof}
Both are immediate from the definition $R{\uparrow}P := \bigcup_{t\in R} t{\uparrow}P$.
\end{proof}

A transaction and a transition are structurally identical---both are pairs of configurations---but differ in their role: a transaction is specified over its participants $Q$, the agents whose states are preconditions for the transaction to occur, and says nothing about agents outside $Q$; different transactions may have different sets of participants.  A transition, by contrast, 
is over a fixed set of agents $P$, as it is a building block of a transition system over $P$ that consists of transitions over $P$.  Given a set of transactions, each over its own set of participants, the closure operator induces from them a set of transitions over a fixed~$P$, in which non-participants remain stationary.

A set of machine transactions $R$ over $S$, each with participants $Q\subseteq P$, defines a multiagent transition system over $S$ and $P$ as follows:

\begin{definition}[Transactions-Based Multiagent Transition System]\label{definition:tbdts}
Given agents $P \subset \Pi$,  states $S$ with initial state $s_0\in S$, 
and a set of transactions $R$, each $t\in R$ over some $Q\subseteq P$ and  $S$, the \temph{transactions-based multiagent transition system} over $P$, $S$, and $R$ is the multiagent transition system $TS= (S^P,\{s_0\}^P,R{\uparrow}P,\emptyset)$, with $\sim$ the empty partial equivalence, so that no transition belongs to a class and no liveness obligation is imposed.
\end{definition}
In other words, one can fully specify a multiagent transition system over $S$ and $P$ simply by providing a set of transactions over $S$, each with participants $Q\subseteq P$.  

Similarly, a set $R$ of volition-guarded machine transactions with an equivalence $\sim$ on their underlying machine transactions induces a volitional multiagent transition system:

\begin{definition}[Volitional Multiagent Transition System]\label{definition:vmts}
Given agents $P\subset \Pi$, machine states $S$ with initial state $s_0$, a set $R$ of volition-guarded machine transactions such that every $(t,Q')\in R$ has the participants of $t$ contained in $P$ and distinct underlying machine transactions have disjoint $P$-closures, and an equivalence $\sim$ on the set $T_R := \{t : (t,Q')\in R\text{ for some }Q'\}$ of underlying machine transactions (a total equivalence, hence in particular a partial equivalence in the sense of Definition~\ref{definition:ts}), the \temph{volitional multiagent transition system induced by $(S,R,{\sim})$ over $P$} is the multiagent transition system $(\calA^P,c_0,T_V,{\sim_V})$ where:
\begin{enumerate}
    \item $\calA := 2^{T_R/\sim} \times S$ is the \temph{agent state space};
    \item $c_0 \in \calA^P$ is the \temph{initial agent configuration}, with $c_0{^v_p}=\emptyset$ and $c_0{^m_p}=s_0$ for every $p\in P$;
    \item $T_V$ consists of all transitions $e\to e'\in (\calA^P)^2$ of one of two forms: (\ia)~a \temph{change-volition} of some $p\in P$---$e^v_p, e'{^v_p}\subseteq (T_R/{\sim})_p$ and $e^v_p \ne e'{^v_p}$, $e^m_p = e'{^m_p}$, and $e_r = e'_r$ for every $r\in P\setminus\{p\}$; or (\ib)~a \temph{volitional machine transaction} induced by some volition-guarded machine transaction $(t,Q')\in R$ per Definition~\ref{definition:vmat}(2);
    \item $\sim_V$ is the restriction of ${\sim}{\uparrow}P$ (Definition~\ref{definition:closure}) to the volitional machine transactions, relating two of them whenever their inducing machine transactions are $\sim$-equivalent, and leaves every change-volition transition outside its domain---so change-volitions belong to no class and carry no liveness obligation.
\end{enumerate}
\end{definition}
The set $R$ of volition-guarded machine transactions that specifies a platform thus determines the volitional machine transitions of the induced VMTS, while change-volition transitions are freely available to every agent independent of $R$.

We define when a volition-guarded transaction is enabled, combining the machine precondition with the volitions of its guards.

\begin{definition}[Enabled]\label{definition:enabled}
Given a set of volition-guarded machine transactions, each $(t,Q')$ with $t = d \rightarrow d'$ a machine transaction over some  $Q' \subseteq Q\subseteq P$ and $S$, and an equivalence $\sim$ on machine transactions: the volition-guarded transaction $(t,Q')$ is \temph{enabled} in agent configuration $c$ over $P$ if $c^m_p = d_p$ for every $p \in Q$, and $[t] \in c^v_q$ for every $q \in Q'$.  An equivalence class $[t]$ is \temph{enabled} in $c$ if some volition-guarded $(t',Q')$ with $t' \in [t]$ is enabled in $c$.
\end{definition}
A volition-guarded transaction with an empty guard ($Q' = \emptyset$) requires no volitions and is enabled whenever its machine precondition is met.  For the volitional multiagent transition system above, a class of volitional machine transactions is enabled at a configuration in the sense of Definition~\ref{definition:ts} exactly when some volition-guarded transaction inducing it is enabled in the above sense; this determines which runs are live and correct per Definition~\ref{definition:ts}.  Change-volition transitions belong to no class and so impose no liveness obligation; personal choices remain free.

\subsection{Implementations}\label{subsection:implementations}

An implementation relates a specification transition system to one that implements it, via a mapping of states.  We adapt the notion of implementation among transition systems~\cite{shapiro2021multiagent,shapiro2026implementing}, with safety and liveness as defined above.

\begin{definition}[Implementation]\label{definition:implementation}
Given transition systems $TS=(S,s_0,T,{\sim})$, the \temph{specification}, and $TS'=(S',s'_0,T',{\sim'})$, an \temph{implementation of $TS$ by $TS'$} is a function $\sigma:S'\to S$ with $\sigma(s'_0)=s_0$, in which case the pair $(TS',\sigma)$ is an \temph{implementation of $TS$}.  Given a computation $r'=s'_1,s'_2,\ldots$ of $TS'$, $\sigma(r')$ is the computation $\sigma(s'_1),\sigma(s'_2),\ldots$ with consecutive repetitions removed; a transition $s'\to d'$ of $TS'$ with $\sigma(s')=\sigma(d')$ is a \temph{stutter}, contributing no step to $\sigma(r')$.
\end{definition}

\begin{definition}[Correct and Complete Implementation]\label{definition:correct-implementation}
An implementation $(TS',\sigma)$ of $TS$ is \temph{correct} if $\sigma$ maps every correct run of $TS'$ to a correct run of $TS$, and \temph{complete} if every correct run $r$ of $TS$ is $\sigma(r')$ for some correct run $r'$ of $TS'$.
\end{definition}

The original framework~\cite{shapiro2021multiagent} also accounts for faulty behaviour and its resilience, which we adapt to the present setting.

\begin{definition}[Safety Fault, Fault-Resilient Implementation]\label{definition:fault-resilient}
Let $(TS',\sigma)$ be an implementation of $TS$, with $TS'=(S',s'_0,T',{\sim'})$.  A \temph{safety fault} is a set $F\subseteq S'^2\setminus T'$ of \temph{faulty transitions}, and a computation \temph{performs} $F$ if it includes a transition in $F$.  The implementation is \temph{$F$-resilient} if $\sigma$ maps every live run $r'\subseteq T'\cup F$ of $TS'$ to a correct run of $TS$.
\end{definition}
Faulty transitions lie outside $T'$, hence in no class of $\sim'$, so they impose no liveness requirement.

The dual of a safety fault is a \temph{liveness fault}: rather than acting incorrectly, correct agents cease to act, so that obligations that ought to be met are not.  A safety fault enlarges the transitions that may occur; a liveness fault waives some of the obligations that liveness would otherwise impose.

\begin{definition}[Liveness Fault, Liveness-Fault-Resilient Implementation]\label{definition:liveness-fault}
Let $(TS',\sigma)$ be an implementation of $TS$, with $TS'=(S',s'_0,T',{\sim'})$.  A \temph{liveness fault} is a set $\Lambda\subseteq T'/{\sim'}$ of liveness classes.  A run \temph{performs} $\Lambda$ if it is not live with respect to some class in $\Lambda$.  A run is \temph{$\Lambda$-live} if no class $[t]\in T'/{\sim'}\setminus\Lambda$ is enabled in every state of some suffix with no member of $[t]$ occurring in the suffix---that is, live for $\sim'$ with the classes in $\Lambda$ removed from its domain.  The implementation is \temph{$\Lambda$-resilient} if $\sigma$ maps every $\Lambda$-live run of $TS'$ to a correct run of $TS$.
\end{definition}
With $\Lambda=\emptyset$ the $\Lambda$-live runs are exactly the live runs.  In a volitional multiagent transition system the canonical liveness fault is \temph{fail-stop}: when a set of agents $D\subseteq P$ cease to act, the waived classes are $\Lambda_D := \{[t]\in T'/{\sim'} : \text{some participant of } t \text{ is in } D\}$, every class some member of which requires a stopped agent.  This is well-defined because $\sim'$-equivalent transactions share participants (Definition~\ref{definition:equivalence}); and it respects guards, since a class guarded by either of two persons retains an obligated representative as long as one of them is live.

The instance-level notion of resilient implementation that composes along a protocol stack is developed in~\cite{shapiro2021multiagent} and can be adapted to this setting.

\section{Grassroots Protocols}\label{section:gs-protocols}

Here we define what is a protocol; 
define when a protocol is grassroots, using the notion of interleaving of correct runs;
show how to define a protocol via a set of transactions; 
prove that any transactions-based protocol is oblivious under a natural condition;
and conclude that if it is also interactive, it is grassroots.
For volitional transactions-based protocols, guarding every transaction of two or more participants guarantees this condition; and a stronger notion, volitionally grassroots, requires the first interaction between any two groups to be willed by a member of each.

\subsection{Protocols and Grassroots Protocols}

A protocol is a family of multiagent transition systems, one for each set of agents $P\subset \Pi$, which share an underlying set of local machine states $\calS$ with a designated initial state $s_0$.
A \emph{local-states function} maps every set of agents $P \subset \Pi$ to an arbitrary set of local machine states $S(P)\subset \calS$ that includes $s_0$ and satisfies $P\subset P' \subset \Pi \implies S(P) \subset S(P')$.
\begin{definition}[Protocol]\label{definition:family}
A \temph{protocol} $\calF$ over a local-states function $S$ is a family of multiagent transition systems that has exactly one transition system $\calF(P) = (C(P),c_0(P),T(P),{\sim(P)})$ for every $P \subset \Pi$, with \temph{agent states} $\calA(P)$, configurations $C(P) := \calA(P)^P$, initial configuration $c_0(P)\in C(P)$, and partial equivalence $\sim(P)$ on $T(P)$ determined by the protocol, such that $P\subseteq P' \subset \Pi$ implies $\calA(P)\subseteq \calA(P')$ and $c_0(P)_p = c_0(P')_p$ for every $p\in P$.
\end{definition}

Informally, in a grassroots protocol two disjoint groups of agents can each operate independently---their interleaved correct runs are correct runs of the combined system---yet the combined system offers genuinely new behaviours that neither group could produce on its own.  To capture this notion formally, we first define the interleaving of runs of two disjoint groups.

\begin{definition}[Interleaving]\label{definition:interleaving}
Let $P, P' \subset \Pi$ be disjoint nonempty sets of agents, $r = c_0, c_1, \ldots$ a run of $\calF(P)$, and $r' = d_0, d_1, \ldots$ a run of $\calF(P')$.  An \temph{interleaving} of $r$ and $r'$ is a sequence $e_0, e_1, \ldots$ of configurations in $C(P \cup P')$ for which there exist non-decreasing sequences of indices $(i_k)_{k \geq 0}$ and $(j_k)_{k \geq 0}$ with $i_0 = j_0 = 0$ such that for every $k \geq 0$:
\begin{enumerate}
\item $(e_k)_p = (c_{i_k})_p$ for every $p \in P$,
\item $(e_k)_q = (d_{j_k})_q$ for every $q \in P'$,
\item if $e_{k+1}$ exists, then exactly one of:
  (\ia) $i_{k+1} = i_k + 1$ and $j_{k+1} = j_k$ (a $P$-step), or
  (\ib) $i_{k+1} = i_k$ and $j_{k+1} = j_k + 1$ (a $P'$-step).
\end{enumerate}
Moreover, if $r$ is finite of length $n$ then $i_k = n$ for some $k$, and if $r$ is infinite then for every $m \ge 0$ there is a $k$ with $i_k = m$; likewise for $r'$ and $(j_k)$.
\end{definition}
Note that an interleaving is well-defined: by Definition~\ref{definition:family}, $\calA(P) \subseteq \calA(P\cup P')$ and $\calA(P') \subseteq \calA(P\cup P')$, so each $e_k$, with $p$-components in $\calA(P)$ and $q$-components in $\calA(P')$, is a valid configuration in $C(P\cup P') = \calA(P\cup P')^{P\cup P'}$.  Also, $e_0 = c_0(P\cup P')$, since $(e_0)_p = (c_0)_p = c_0(P\cup P')_p$ for $p\in P$ and $(e_0)_q = (d_0)_q = c_0(P\cup P')_q$ for $q\in P'$, by the agreement of initial configurations across $\calF(P)$, $\calF(P')$, and $\calF(P\cup P')$.

\begin{definition}[Interaction, Interactive Run, First Interaction]\label{definition:first-interaction}
Let $P, P'\subset\Pi$ be disjoint and nonempty.  For a configuration $c$ over $P\cup P'$ and $Q\subseteq P\cup P'$, write $c|_Q$ for the restriction of $c$ to $Q$.  A transition $e\to e'$ of $\calF(P\cup P')$ is an \temph{interaction between $P$ and $P'$} if $e|_P \ne e'|_P$ and $e|_{P'} \ne e'|_{P'}$, and $e|_P \to e'|_P$ is not a transition of $\calF(P)$ or $e|_{P'} \to e'|_{P'}$ is not a transition of $\calF(P')$.  A run $\hat r$ of $\calF(P\cup P')$ is \temph{interactive between $P$ and $P'$} if some transition of $\hat r$ is an interaction between $P$ and $P'$; the \temph{first interaction} of $P$ and $P'$ in $\hat r$ is the first such transition.
\end{definition}

We can now define the key notion of this paper, a grassroots protocol.  The following definition improves upon the original definition~\cite{shapiro2023grassrootsBA}, which was formulated in terms of a subset relation ($P\subset P'$) and conditions on the availability of transitions, and upon the definition of~\cite{shapiro2025atomic}, which did not incorporate liveness.  The new definition captures the informal notion of grassroots directly, using disjoint groups and the interleaving of their correct runs; the differences are discussed below.

\begin{definition}[Oblivious, Interactive, Grassroots]\label{definition:grassroots}
A  protocol $\calF$ is:
\begin{enumerate}
    \item \temph{oblivious} if for every disjoint nonempty $P, P' \subset \Pi$,
    every interleaving of a correct run of $\calF(P)$ and a correct run of $\calF(P')$ is a correct run of $\calF(P\cup P')$.
    \item  \temph{interactive} if for every disjoint nonempty $P, P' \subset \Pi$, some correct run of $\calF(P\cup P')$ is interactive between $P$ and $P'$.
    \item \temph{grassroots} if it is oblivious and interactive.
\end{enumerate}
\end{definition}

\mypara{Oblivious} Being oblivious means two disjoint groups coexist without interference: any interleaving of a correct run of one with a correct run of the other is a correct run of the combined system, so each group's correct operation is undisturbed by the other's presence.

\mypara{Interactive} Being interactive means the combined system has a correct run in which a single transaction engages agents of both groups---a coupling that no interleaving of the groups' separate runs can produce.  The first such transaction in a run is its first interaction (Definition~\ref{definition:first-interaction}).  The substance of interactivity is the content of these transactions, which the platforms in Section~\ref{section:platforms} illustrate.

Federated systems such as Mastodon~\cite{raman2019challenges} are oblivious, as servers in one group can ignore servers in the other, but are not interactive:  a group of clients $P$ without a server cannot do more when joined by another group of clients $P'$, also without a server.
Any protocol that employs a shared global data structure---whether replicated (Blockchain~\cite{bitcoin}) or distributed (DHT~\cite{maymounkov2002kademlia}, IPFS~\cite{benet2014ipfs})---is not oblivious, and hence not grassroots; the argument covers trackerless BitTorrent~\cite{torrentfreak2021bittorrent}, which relies on a DHT.  In Appendix~\ref{appendix:not-grassroots} we prove that Bitcoin is not grassroots and show that the same argument applies to distributed hash tables and IPFS.

\subsection{Transactions-Based Grassroots Protocols}

A protocol can be specified by a set of machine transactions over a local-states function, lifting the closure construction of Section~\ref{section:dts} from a fixed set of agents to the whole family.  We show that such a protocol is oblivious under a natural condition, and hence grassroots when it is also interactive.

\begin{definition}[Transactions Over a Local-State Function]\label{definition:tblsf}
Let $S$ be a local-states function.
A set of transactions $R$ is \temph{over $S$} if every transaction $t\in R$
is a multiagent transition over $Q$ and $S(P')$ for some $Q \subseteq P'\subset \Pi$.  Given such a set $R$ and $P\subset \Pi$,
$
R(P) := \{ t\in R : t \text{ is over } Q \text{ and } S(P'), Q \subseteq P'\subseteq P\}
$.
\end{definition}

\begin{definition}[Transactions-Based Protocol]\label{definition:mats-protocol}
Let $S$ be a local-states function and $R$ a set of machine transactions over $S$, in which distinct transactions have disjoint $P$-closures for every $P\subset\Pi$, with equivalence $\sim$.  The \temph{transactions-based protocol over $R$, $S$, and $\sim$} assigns to each $P\subset \Pi$ the multiagent transition system $\calF(P) := (S(P)^P,\{s_0\}^P,R(P){\uparrow}P,{\sim}{\uparrow}P)$, where $R(P){\uparrow}P$ and ${\sim}{\uparrow}P$ are the $P$-closures (Definition~\ref{definition:closure}) of $R(P)$ and of $\sim$ restricted to $R(P)$.  In particular $\calA(P) = S(P)$, monotone as Definition~\ref{definition:family} requires, $C(P) = S(P)^P$, and $c_0(P) = \{s_0\}^P$.
\end{definition}

By Definition~\ref{definition:mats-protocol}, the transitions of $\calF(P\cup P')$ are induced by the machine transactions in $R(P\cup P')$.  Since liveness ranges over the equivalence classes of $\calF(P\cup P')$ (Definition~\ref{definition:liveness}), any class $[t]$ in $R(P\cup P')/\!\sim$ whose transactions have participants spanning both $P$ and $P'$ could obstruct obliviousness if enabled in an interleaving.  The following proposition identifies the condition under which this does not occur.

\begin{proposition}\label{proposition:oblivious}
A transactions-based protocol is oblivious provided that for every disjoint nonempty $P, P' \subset \Pi$,  no equivalence class whose transactions have participants spanning both $P$ and $P'$ is ever enabled in any interleaving of correct runs of $\calF(P)$ and $\calF(P')$.
\end{proposition}

\begin{proof}
Let $\calF$ be a transactions-based protocol over machine transactions $R$, local-states function $S$, and equivalence $\sim$.  Let $P, P' \subset \Pi$ be disjoint and nonempty, $r = c_0, c_1, \ldots$ a correct run of $\calF(P)$, $r' = d_0, d_1, \ldots$ a correct run of $\calF(P')$, and $e = e_0, e_1, \ldots$ an interleaving of $r$ and $r'$.

\mypara{Safety}  $e_0 = c_0(P\cup P')$ as noted above.  Consider a $P$-step $e_k \rightarrow e_{k+1}$; the transition $c_{i_k} \rightarrow c_{i_{k+1}}$ of $\calF(P)$ it lifts lies in $R(P){\uparrow}P$, hence in $t{\uparrow}P$ for some $t\in R(P)$ over participants $Q\subseteq P$.  In $e_k\to e_{k+1}$ the agents of $Q$ do what they do in $t$ (matching $c_{i_k}\to c_{i_{k+1}}$ on $Q$), the agents of $P\setminus Q$ are stationary, and the agents of $P'$ are stationary because it is a $P$-step.  Every agent outside $Q$ is thus stationary and $Q$ matches $t$, so $e_k\to e_{k+1}\in t{\uparrow}(P\cup P')$.  Since $R(P)\subseteq R(P\cup P')$, $e_k\to e_{k+1}\in R(P\cup P'){\uparrow}(P\cup P')$, a transition of $\calF(P\cup P')$.  The case of a $P'$-step is symmetric.

\mypara{Liveness}  Suppose for contradiction that some class $[t]\in R(P\cup P')/\!\sim$ is enabled in some suffix of $e$ with no member of $[t]$ taken in the suffix.  By the hypothesis of the Proposition, every representative of $[t]$ enabled at any $e_k$ has participants contained in $P$ or contained in $P'$.  Consider a representative $t'$ enabled at some $e_k$ in the suffix; without loss of generality its participants $Q\subseteq P$.  Enablement of $[t]$ at $e_k$ depends only on the states of agents in $P$, which match those of $r$ at index $i_k$.  Since the interleaving exhausts $r$, the indices $i_k$ over the suffix cover a tail of $r$, so $[t]$ is enabled at every configuration of that tail; as no member of $[t]$ is taken in the suffix, none is taken along that tail of $r$.  This contradicts correctness of $r$.  The $P'$ case is symmetric.
\end{proof}

\mypara{Relation to the original grassroots definition}
The original definition of grassroots protocols~\cite{shapiro2023grassrootsBA,shapiro2025atomic} used a subset relation ($P\subset P'$) and conditions on transition availability, without incorporating liveness.  The new definition uses disjoint groups and interleavings of correct runs.  The two definitions are not comparable in general---neither implies the other---but we informally claim that, for the platforms considered here, both definitions agree.

\subsection{Volitional Transactions-Based Grassroots Protocols}\label{subsection:volitional-protocols}

A volitional transactions-based protocol assigns to each set of agents the volitional multiagent transition system induced by its volition-guarded transactions.  Obliviousness follows from a syntactic guard condition, in three steps.  First, the Volitional Containment Lemma (Lemma~\ref{lemma:volitional-containment}) establishes that an agent's volitional state, in any configuration of $\calF(P)$, never contains classes of transactions outside $R(P)$ (Definition~\ref{definition:tblsf}).  Second, Proposition~\ref{proposition:volitional-oblivious} uses this invariant to reduce obliviousness to a condition on equivalence classes in the interleaving.  Third, Corollary~\ref{corollary:guarded-oblivious} establishes that condition whenever every volition-guarded transaction of two or more participants has a nonempty guard.  A protocol so oblivious is grassroots once it is interactive (Definition~\ref{definition:grassroots}).  We then strengthen grassroots to a volitional notion: volitionally grassroots additionally requires the first interaction between any two groups to be willed by a member of each.

Throughout this subsection, $\sim$ is an equivalence on $T_R$, the machine transactions underlying the whole set $R$; for $T'\subseteq T_R$ we write $T'/{\sim}$ for the set of $\sim$-classes with a representative in $T'$, and $[t]$ for the class of $t$ in $T_R/{\sim}$.  Hence $P\subseteq P'$ implies $T_{R(P)}/{\sim}\subseteq T_{R(P')}/{\sim}$.

\begin{definition}[Volitional Transactions-Based Protocol]\label{definition:protocol-transactions}
Let $S$ be a local-states function and $R$ a set of volition-guarded transactions over $S$ with equivalence $\sim$.
The \temph{protocol $\calF$ over $R$, $S$, and $\sim$} assigns to each set of agents $P\subset \Pi$ the volitional multiagent transition system $\calF(P)$ induced by $(S(P),R(P),{\sim})$ over $P$ (Definition~\ref{definition:vmts}).  In particular, $C(P) = \calA(P)^P$ with agent state space $\calA(P) := 2^{T_{R(P)}/\sim} \times S(P)$---monotone in $P$, as Definition~\ref{definition:family} requires---and $c_0(P)$ has $p$-component $(\emptyset,s_0)$ for every $p\in P$.
\end{definition}

The obliviousness results below rest on the following invariant: in a group's own system an agent wills only transactions internal to the group, so the guard of a transaction that reaches another group cannot will it while the group runs alone, and the transaction is never enabled in an interleaving.

\begin{lemma}[Volitional Containment]\label{lemma:volitional-containment}
Let $\calF$ be a volitional transactions-based protocol over a set of volition-guarded transactions $R$ with equivalence $\sim$.  For every $P\subset \Pi$, every configuration $c$ of $\calF(P)$, and every $p\in P$: $c^v_p \subseteq T_{R(P)}/\!\sim$, where $T_{R(P)} := \{t : (t,Q')\in R(P)\text{ for some }Q'\}$.
\end{lemma}
\begin{proof}
By Definition~\ref{definition:vmts}, every configuration of $\calF(P)$ is in $\calA(P)^P$ with $\calA(P) = 2^{T_{R(P)}/\sim}\times S(P)$, so $c^v_p \subseteq T_{R(P)}/\!\sim$ for every $p\in P$.
\end{proof}

\begin{proposition}\label{proposition:volitional-oblivious}
A volitional transactions-based protocol is oblivious provided that for every disjoint nonempty $P, P' \subset \Pi$, no equivalence class whose transactions have participants spanning both $P$ and $P'$ is ever enabled in any interleaving of correct runs of $\calF(P)$ and $\calF(P')$.
\end{proposition}
\begin{proof}
Let $\calF$ be a volitional transactions-based protocol over volition-guarded transactions $R$ with equivalence $\sim$.  Let $P, P' \subset \Pi$ be disjoint and nonempty, $r = c_0, c_1, \ldots$ a correct run of $\calF(P)$, $r' = d_0, d_1, \ldots$ a correct run of $\calF(P')$, and $e = e_0, e_1, \ldots$ an interleaving of $r$ and $r'$.

\mypara{Safety}  $e_0 = c_0(P\cup P')$ as noted above.  A $P$-step $e_k \rightarrow e_{k+1}$ lifts a transition $c_{i_k} \rightarrow c_{i_{k+1}}$ of $\calF(P)$ that, by Definition~\ref{definition:vmts}, is a change-volition of some $p\in P$ or a volitional machine transaction induced by some $(t,Q')\in R(P)$.  A change-volition of $p\in P$ is a change-volition of $p\in P\cup P'$ in $\calF(P\cup P')$, lifted with $P'$-agents unchanged because it is a $P$-step.  A volitional machine transaction induced by $(t,Q')\in R(P)$, with $t=d\to d'$ over $Q\subseteq P$, is induced by $(t,Q')\in R(P\cup P')$: its machine and volitional preconditions (Definition~\ref{definition:vmat}(2)) hold at $e_k$ because they hold at $c_{i_k}$, and $P'$-agents are machine-stationary because it is a $P$-step.  The postcondition $e_{k+1}{^v_p} = e_k{^v_p}\setminus\{[t]\}$ holds on $P$-agents by the $\calF(P)$-step; on each $p'\in P'$, by Lemma~\ref{lemma:volitional-containment} applied to $r'$, $e_k{^v_{p'}} = (d_{j_k}){^v_{p'}}\subseteq T_{R(P')}/\!\sim$, and $t\notin R(P')$ since $Q\subseteq P$ and $Q\cap P'=\emptyset$, so $[t]\notin e_k{^v_{p'}}$ and the postcondition holds vacuously.  Hence $e_k \rightarrow e_{k+1}$ is a transition of $\calF(P\cup P')$.  The case of a $P'$-step is symmetric.

\mypara{Liveness}  By the Liveness argument of Proposition~\ref{proposition:oblivious}, $e$ is live.
\end{proof}

When every volition-guarded transaction of two or more participants has a nonempty guard, the hypothesis of Proposition~\ref{proposition:volitional-oblivious} is established by a uniform argument about volitions, without invoking platform-specific machine preconditions.

\begin{corollary}[Volition-Guarded Obliviousness]\label{corollary:guarded-oblivious}
A volitional transactions-based protocol over a set of volition-guarded transactions $R$ is oblivious if every volition-guarded transaction $(t,Q')\in R$ whose machine transaction $t$ has two or more participants has a nonempty guard, $Q'\ne\emptyset$.
\end{corollary}
\begin{proof}
We verify the hypothesis of Proposition~\ref{proposition:volitional-oblivious}.  Let $P, P'\subset\Pi$ be disjoint and nonempty, $e$ an interleaving of a correct run of $\calF(P)$ and a correct run of $\calF(P')$, $[t]$ an equivalence class of $R(P\cup P')/{\sim}$ whose participants $Q$ span both $P$ and $P'$, and $(t,Q')$ a representative.  Since $Q$ has participants in both groups, $|Q|\ge 2$, so by hypothesis $Q'\ne\emptyset$; pick $q\in Q'$ and, without loss of generality, $q\in P$.  Since $Q$ has a participant in $P'$, $t\notin R(P)$ by Definition~\ref{definition:tblsf}; and by well-formedness of $\sim$ (Definition~\ref{definition:equivalence}), every $t'\sim t$ has the same participants as $t$, so $[t]\notin T_{R(P)}/\!\sim$.  By Lemma~\ref{lemma:volitional-containment} applied to the $P$-run, $e_k{^v_q} \subseteq T_{R(P)}/\!\sim$ at every $e_k$, so $[t]\notin e_k{^v_q}$ and the guard on $q$ fails.  Hence $(t,Q')$ is not enabled at any $e_k$.
\end{proof}

Corollary~\ref{corollary:guarded-oblivious} reduces obliviousness---for volitional transactions-based protocols---to a syntactic check on guards.  The two platforms of Section~\ref{section:platforms} apply it uniformly.

In a run of $\calF(P\cup P')$ a member may will a transaction reaching the other group, which a group's own run forbids; the volitional notion constrains the resulting coupling:

\begin{definition}[Volitionally Grassroots]\label{definition:volitionally-grassroots}
A volitional transactions-based protocol $\calF$ is \temph{volitionally grassroots} if it is grassroots and, for every disjoint nonempty $P, P'\subset\Pi$ and every safe run of $\calF(P\cup P')$ interactive between $P$ and $P'$, the first interaction of $P$ and $P'$ is induced by a volition-guarded transaction $(t,Q')$ with $Q'\cap P\ne\emptyset$ and $Q'\cap P'\ne\emptyset$.
\end{definition}

\section{Grassroots Platforms via Volition-Guarded Multiagent Atomic Transactions}\label{section:platforms}

We now present two grassroots platforms: grassroots social networks  and grassroots coins and bonds.  For each platform, we specify its volition-guarded transactions, define the induced volitional multiagent transition system, prove platform-specific invariants, and prove the platform is grassroots via Corollary~\ref{corollary:guarded-oblivious} and Definition~\ref{definition:grassroots}.  By \emph{invariants} we mean properties preserved across all runs of the induced multiagent transition system; these are the platform-specific analogue of safety in the classical sense.  In both platforms every participant of every transaction changes state, so distinct transactions have disjoint closures.

\subsection{Grassroots Social Networks via Befriending and Unfriending}\label{subsection:social-networks}

In a grassroots social network~\cite{shapiro2023gsn}, the social graph is stored distributively under the control of the people themselves, with each person storing the local neighbourhood pertaining to them, and no third-party having access unless explicitly granted.  The original definition~\cite{shapiro2023gsn} was via a unary multiagent transition system; here both actions are specified as binary transactions.

Each agent $p$ maintains, as its local state, a finite set $c_p\subseteq P$ recording the friends of $p$; initially $c_p = \emptyset$.  Befriending adds $q$ to $c_p$ and $p$ to $c_q$; unfriending removes them.  Communication functions of a social network can be added, under the restriction that communication occurs only among friends~\cite{shapiro2023gsn}. A liveness theorem can be proven for this design~\cite{shapiro2023gsn}, stating that if a person $p$ that follows a person $q$ is connected to $q$ via a chain of mutual friends, each of them correct and follows $q$, then $p$ will eventually receive every item on $q$'s feed.
The specification of the grassroots social graph is the foundation for grassroots social networks with feeds, groups, messaging, explored elsewhere~\cite{shapiro2023gsn,shapiro2025glp,shapiro2026volitional}.

\begin{definition}[Grassroots Social Graph]\label{definition:SocialGraph}
The \temph{grassroots social graph} $SG$ is the protocol over the befriend and unfriend volition-guarded transactions of the Introduction, with local-states function $S(P) := 2^P$ and equivalence $\sim$ that identifies all befriend$(p,q)$ transactions with each other and all unfriend$(p,q)$ transactions with each other, per Definition~\ref{definition:protocol-transactions}.
\end{definition}

\mypara{Invariants}
\begin{lemma}[Mutuality]\label{lemma:friendship-safety}
Given a safe run $r$ of $SG$, a configuration $c\in r$, and agents $p, q\in P$,
$q \in c_p \iff p \in c_q$. 
\end{lemma}
\begin{proof}
By induction on the length of the run $r=c_0,c_1,\ldots,c_n$.  In the initial configuration $c_0 = \emptyset^P$ the biconditional holds vacuously.
Assume the lemma holds for $c_n$, and consider the transition $c_n\rightarrow c_{n+1}$.  It can be either Befriend or Unfriend for some pair $\{p,q\}$; both modify $c_p$ and $c_q$ symmetrically---Befriend adds $q$ to $c_p$ and $p$ to $c_q$, Unfriend removes them---preserving the biconditional.  For all other pairs $\{r,s\}$, the local states are unchanged.
\end{proof}

We note that each configuration $c$ in a safe run $r$ of $SG$ induces a graph with agents as vertices and an edge $p\leftrightarrow q$ when $q \in c_p$ (equivalently, by Lemma~\ref{lemma:friendship-safety}, when $p \in c_q$), and that the graphs induced by two consecutive configurations in $r$ differ by exactly one added or removed edge.

\mypara{Transaction equivalence} All befriend$(p,q)$ transactions---differing in the configurations in which they occur---form an equivalence class.  Similarly for unfriend$(p,q)$.

\mypara{Liveness}  Liveness applies to the befriend and unfriend classes (Definition~\ref{definition:liveness}).  Unfriending is guarded by either $p$ or $q$: once either person wills the class, the transaction becomes enabled and must eventually be taken.  Befriending is guarded by both $p$ and $q$: it becomes enabled only when both persons will the class, and must then eventually be taken.

\mypara{Grassroots}
\begin{corollary}\label{corollary:GSN}
The grassroots social graph $SG$ is volitionally grassroots.
\end{corollary}
\begin{proof}
\emph{Grassroots:}  Befriend and unfriend are guarded---befriend by $\{p,q\}$, unfriend by $p$ or by $q$---so every volition-guarded transaction of two or more participants has a nonempty guard, and by Corollary~\ref{corollary:guarded-oblivious}, $SG$ is oblivious.  For disjoint nonempty $P, P'$, take $p\in P$ and $q\in P'$ and the run $\hat{r}$ of $SG(P\cup P')$ in which $p$ and $q$ each will befriend and the befriend is then taken; the befriend changes the local states of both $p$ and $q$, so $\hat{r}$ is interactive, and $\hat{r}$ is correct, as the befriend fulfils $[\mathrm{befriend}(p,q)]$ in both volitions and no class is enabled at its end.  Thus $SG$ is oblivious and interactive, hence grassroots by Definition~\ref{definition:grassroots}.

\emph{Volitionally grassroots:}  By Definition~\ref{definition:first-interaction}, the first interaction of disjoint nonempty $P, P'$ in any safe interactive run changes the state of an agent in each group; in $SG$ it is therefore a befriend or an unfriend between an agent of $P$ and one of $P'$.  An unfriend of $p\in P$ and $q\in P'$ requires $q\in c_p$, an existing friendship that only an earlier befriend of $p$ and $q$---itself an interaction---can create; so the first interaction is not an unfriend.  It is therefore a befriend, guarded by both its participants, one in $P$ and one in $P'$.  Hence $SG$ is volitionally grassroots.
\end{proof}

\subsection{Grassroots Coins and Bonds}\label{subsection:bonds}

Grassroots coins~\cite{shapiro2024gc,lewis2023grassroots} are units of debt that can be issued and traded digitally by any person.  Each person's coins are backed by the goods and services they offer, priced in their own currency, with liquidity arising from mutual credit via coin exchange among persons that know and trust each other.  Grassroots bonds~\cite{shapiro2026bonds} extend grassroots coins with a maturity date, reframing grassroots coins---cash---as mature grassroots bonds.  Coin-for-bond redemption generalises coin-for-coin redemption, allowing the lending of liquid coins in exchange for interest-bearing future-maturity bonds.  Digital social contracts---voluntary agreements among persons, specified, fulfilled, and enforced digitally---can express the full gamut of financial instruments as the voluntary swap of grassroots bonds, including credit lines, loans, sale of debt, forward contracts, options, and escrow-based instruments~\cite{shapiro2026bonds}.

\begin{definition}[Grassroots Bonds]\label{definition:bonds}
A \temph{$p$-bond with maturity date $d$}, denoted \textcent$_{p,d}$, is a unit of debt issued by $p \in \Pi$ maturing at date $d\in \calN$.  We let $\calB(P) =\{$\textcent$_{p,d} : p\in P, d\in \calN\}$ denote the set of all grassroots bonds by agents $P\subset\Pi$.  Each agent $p$ maintains as its local state a pair $(c_p, d_p^*)$ where $c_p$ is a multiset of members of $\calB(P)$ (initially $\emptyset$) and $d_p^*\in \calN$ is the local current date (initially $0$); $p$ considers a bond \textcent$_{q,d}$ to be \temph{mature}, and refers to it as a \temph{$q$-coin} (denoted \textcent$_q$), iff $d\le d_p^*$.  There is no global date; agents may disagree on which bonds are mature.
\end{definition}

The grassroots bonds volition-guarded transactions are:
\begin{tcolorbox}[colback=gray!5!white,colframe=black!75!black,top=2pt,bottom=2pt]
\begin{enumerate}
    \item \textbf{Mint}: $c'_p := c_p \cup \text{\textcent}^k_{p,d}$, $k>0$, $d\in\calN$; $d_p^*$ unchanged.  Guarded by $p$.
    \item \textbf{Advance-date}: ${d_p^*}' > d_p^*$; $c_p$ unchanged.  Unguarded.
    \item \textbf{Voluntary swap}: $c'_p := (c_p\cup y)\setminus x$, $c'_q := (c_q\cup x)\setminus y$, provided $x\subseteq c_p$, $y\subseteq c_q$; $d_p^*$ and $d_q^*$ unchanged.  Guarded by $\{p,q\}$.
    \item \textbf{Pay}: $c'_p := c_p\setminus x$, $c'_q := c_q\cup x$, where $x\subseteq c_p$ is a set of $q$-coins (that is, bonds \textcent$_{q,d}$ with $d\le d_p^*$); $d_p^*$ and $d_q^*$ unchanged.  Guarded by $p$.
    \item \textbf{Redeem}: $c'_p := (c_p\cup y)\setminus x$, $c'_q := (c_q\cup x)\setminus y$, where $x = \{\text{\textcent}_{q,d'}\}\subseteq c_p$ with $d'\le d_p^*$, $y = \{\text{\textcent}_{r,d}\}\subseteq c_q$, $r\in P$, $d\in\calN$; $d_p^*$ and $d_q^*$ unchanged.  Guarded by $p$.
\end{enumerate}
\end{tcolorbox}
\noindent Minting, paying, and redeeming are guarded by the initiator; voluntary swap requires both parties to be willing; Advance-date is unguarded, since local time advances mechanically.  In redemption, the redeemer chooses any bond held by the coin's issuer---regardless of who issued the bond---generalising coin-for-coin redemption~\cite{shapiro2024gc} to coin-for-bond redemption~\cite{shapiro2026bonds}.

\begin{definition}[Grassroots Coins and Bonds]\label{definition:gcb}
The \temph{grassroots coins and bonds} $\mathit{GCB}$ is the protocol over the volition-guarded transactions above with local-states function mapping each $P\subset\Pi$ to the set of pairs $(c,d)$ where $c$ is a multiset of members of $\calB(P)$ and $d\in\calN$, and equivalence $\sim$ identifying Mint transactions by the same agent $p$ with the same $k$ and $d$, Advance-date transactions of the same agent, Swap transactions between the same pair exchanging the same multisets, Pay transactions from the same payer to the same payee transferring the same multiset, and Redeem transactions between the same pair exchanging the same bonds, per Definition~\ref{definition:protocol-transactions}.
\end{definition}

\mypara{Invariants}
\begin{lemma}[Conservation of Money]\label{lemma:conservation}
In any safe run $r$ of $\mathit{GCB}$, the $p$-bonds in any configuration $c\in r$ are exactly the $p$-bonds minted by $p$ in the prefix of the run ending in $c$.
\end{lemma}
\begin{proof}
Mint adds new $p$-bonds to $p$'s holdings.  Voluntary swap, Pay, and Redeem transfer bonds between two agents without creating or destroying them: $x$ moves from $p$ to $q$ and $y$ from $q$ to $p$, preserving the total multiset of bonds.  Advance-date changes only $d_p^*$ and leaves bonds unchanged.  Hence the multiset of $p$-bonds across all agents equals the multiset minted by $p$.
\end{proof}

\mypara{Transaction equivalence}
All Mint transactions by the same agent $p$ with the same $k$ and $d$ (differing only in the configurations in which they occur) form an equivalence class.  All Advance-date transactions of the same agent form an equivalence class.  All Swap transactions between the same pair $\{p,q\}$ exchanging the same multisets $x$ and $y$ form an equivalence class; likewise Pay transactions from $p$ to $q$ transferring the same $x$, and Redeem transactions between $p$ and $q$ exchanging the same $x$ and $y$.

\mypara{Liveness}  Mint is guarded by the minting agent; it becomes enabled when the agent wills the class, and must then eventually be taken.  Pay and Redeem are guarded by the initiator; once enabled, they must eventually be taken.  Voluntary swap is guarded by both participants; it becomes enabled only when both persons will the class, and must then eventually be taken.  Advance-date is unguarded and always enabled (since $d_p^*$ can always grow); hence in every correct run it is taken infinitely often for every agent, and $d_p^*$ grows without bound.

\mypara{Grassroots}
\begin{corollary}\label{corollary:GCB}
Grassroots Coins and Bonds are volitionally grassroots.
\end{corollary}
\begin{proof}
\emph{Grassroots:}  Mint, Voluntary swap, Pay, and Redeem are guarded (by $p$; by $\{p,q\}$; by $p$; by $p$), and Advance-date is unary; so every volition-guarded transaction of two or more participants has a nonempty guard, and by Corollary~\ref{corollary:guarded-oblivious}, $\mathit{GCB}$ is oblivious.  For disjoint nonempty $P, P'$, take $p\in P$ and $q\in P'$ and an infinite run $\hat{r}$ of $\mathit{GCB}(P\cup P')$ whose prefix has $p$ mint $p$-coins, $q$ mint $q$-coins, and $p$ and $q$ both will Voluntary swap and execute it (exchanging $p$-coins for $q$-coins), and whose tail takes Advance-date of every agent infinitely often; the swap changes the local states of both $p$ and $q$, so $\hat{r}$ is interactive, and $\hat{r}$ is correct, as Mint and Voluntary swap are fulfilled on execution and never re-willed, and Advance-date is always enabled but taken infinitely often for every agent.  Thus $\mathit{GCB}$ is oblivious and interactive, hence grassroots by Definition~\ref{definition:grassroots}.

\emph{Volitionally grassroots:}  By Definition~\ref{definition:first-interaction}, the first interaction of disjoint nonempty $P, P'$ in any safe interactive run changes the state of an agent in each group; in $\mathit{GCB}$ it is therefore a Voluntary swap, Pay, or Redeem between an agent of $P$ and one of $P'$.  A Pay or Redeem in which $p\in P$ acts towards $q\in P'$ requires $p$ to hold $q$-coins; by Conservation of Money (Lemma~\ref{lemma:conservation}) these are minted by $q$, so they reach $p$ only by an earlier transfer between the two groups---itself an interaction.  So the first interaction is neither a Pay nor a Redeem.  It is therefore a Voluntary swap, guarded by both its participants, one in $P$ and one in $P'$.  Hence $\mathit{GCB}$ is volitionally grassroots.
\end{proof}

\section{Related Work}\label{section:related-work}

\mypara{Atomic transactions}
This work extends the notion of multiagent atomic transactions of~\cite{shapiro2025atomic} with volitions.
Atomic transactions have been investigated early in distributed computing, mostly in the context of database systems~\cite{lampson1981chapter,lynch1993atomic,lynch1988theory}.  Most research since and until today focuses on their efficient and robust implementation~\cite{bravo2019reconfigurable,chockler2021multi}.  The integration of atomic transactions in programming languages has also been explored~\cite{borgstrom2009compositional}.
In terms of formal models of concurrency, the extension of CCS with atomic transactions has been investigated in the past~\cite{acciai2007concurrent,de2010communicating,de2010liveness}, but without follow-on research, so it seems.  While transition systems have been the bedrock of abstract models of computation since the Turing machine, 
we are not aware of previous attempts to explore atomic transactions within their context.

\mypara{Formal models of persons in concurrent systems}
The formal methods tradition has a long lineage of modelling human agents as sources of nondeterminism alongside deterministic machines, but---to the best of our knowledge---without decomposing an agent into person and machine as components of its state.  A detailed survey appears in Appendix~\ref{appendix:persons}; we summarise the key points here.

Turing's choice machines~\cite{turing1936computable} introduced the person as an external operator making free choices at designated states.  Hoare's CSP~\cite{hoare1985communicating} provides process-algebraic encoding via external choice ($\Box$), but models what the person \emph{does}, not what the person \emph{is willing to do}.  Lynch and Tuttle's I/O automata~\cite{lynch1989introduction} model the environment (potentially human) via input-enabling.  Back and von Wright's angelic/demonic nondeterminism~\cite{back1998refinement} is the closest precursor to the volition/obligation distinction: angelic choices model cooperative behaviour, demonic choices model adversarial behaviour.  The distinction from the present work is that angelic nondeterminism is a point-of-choice semantics---a choice is resolved locally at each transition---whereas volitions are persistent, inspectable state, so a guard condition reads the agent's current willing, not a single local resolution.  Game structures~\cite{alur2002alternating,kupferman2001module,abramsky2000full}, ceremony analysis~\cite{ellison2007ceremony}, and electronic institutions~\cite{esteva2001formal,artikis2009specifying} treat agents as symmetric players or norm-governed entities, but none decompose an agent into person and machine components.

Across all traditions, the person is modelled as environment, opponent, error source, or unconstrained nondeterministic process---but always as an entity \emph{external} to the agent, never as a formal component of the agent's state.  The present work, to the best of our knowledge, is the first to decompose each agent's state into a machine state and a volitional state within a multiagent transition system, with machine transactions conditioned on the volitions of the agents' persons, and with liveness arising from the interplay of personal volitions and machine obligations rather than from a designated set of live transitions.

\section{Conclusion and Future Work}\label{section:conclusion}

We have presented volition-guarded multiagent atomic transactions---a formal foundation for describing systems consisting of people operating machines---in which each agent's state decomposes into a machine state and a volitional state, and machine transactions are guarded by their people's volitions.  We have developed the mathematical machinery needed to express the safety and liveness of grassroots platforms thus specified, and demonstrated the framework on two grassroots platforms: social networks and coins and bonds.  We have provided a simpler definition of grassroots that better captures the informal notion and excludes systems based on shared global data structures.

The framework extends naturally to platforms with richer guard structures---such as grassroots federations, in which transactions are guarded by supermajorities or assemblies of community members rather than by all participants---deferred to future work.

The original framework~\cite{shapiro2021multiagent} also considered faulty computations and fault-tolerant implementations.  We have adapted its notions of implementation and fault-resilient implementation to the present setting; establishing fault-tolerant implementations of the specifications presented here is left to future work.

\bibliography{bib}

@article{shapiro2026volitional,
  author    = {Andy Lewis-Pye and Ehud Shapiro},
  title     = {Volitional Multiagent Atomic Transactions: Describing People and their Machines},
  year      = {2026},
  note      = {Submitted, arXiv XXXX.XXXXX}
}

@article{turing1936computable,
  author    = {Alan M. Turing},
  title     = {On Computable Numbers, with an Application to the {E}ntscheidungsproblem},
  journal   = {Proceedings of the London Mathematical Society},
  volume    = {s2-42},
  number    = {1},
  pages     = {230--265},
  year      = {1936}
}

@article{turing1939systems,
  author    = {Alan M. Turing},
  title     = {Systems of Logic Based on Ordinals},
  journal   = {Proceedings of the London Mathematical Society},
  volume    = {s2-45},
  number    = {1},
  pages     = {161--228},
  year      = {1939}
}

@book{milner1980calculus,
  author    = {Robin Milner},
  title     = {A Calculus of Communicating Systems},
  publisher = {Springer},
  series    = {LNCS},
  volume    = {92},
  year      = {1980}
}

@book{milner2009space,
  author    = {Robin Milner},
  title     = {The Space and Motion of Communicating Agents},
  publisher = {Cambridge University Press},
  year      = {2009}
}

@techreport{lynch1989introduction,
  author      = {Nancy A. Lynch and Mark R. Tuttle},
  title       = {An Introduction to Input/Output Automata},
  institution = {MIT Laboratory for Computer Science},
  year        = {1989},
  number      = {MIT/LCS/TM-373}
}

@article{lynch2003hybrid,
  author    = {Nancy Lynch and Roberto Segala and Frits Vaandrager},
  title     = {Hybrid {I/O} Automata},
  journal   = {Information and Computation},
  volume    = {185},
  number    = {1},
  pages     = {105--157},
  year      = {2003}
}

@inproceedings{harel1985development,
  author    = {David Harel and Amir Pnueli},
  title     = {On the Development of Reactive Systems},
  booktitle = {Logics and Models of Concurrent Systems},
  publisher = {Springer},
  series    = {NATO ASI Series},
  volume    = {13},
  pages     = {477--498},
  year      = {1985}
}

@book{back1998refinement,
  author    = {Ralph-Johan Back and Joakim von Wright},
  title     = {Refinement Calculus: A Systematic Introduction},
  publisher = {Springer},
  year      = {1998}
}

@article{kupferman2001module,
  author    = {Orna Kupferman and Moshe Y. Vardi and Pierre Wolper},
  title     = {Module Checking},
  journal   = {Information and Computation},
  volume    = {164},
  number    = {2},
  pages     = {322--344},
  year      = {2001}
}

@article{alur2002alternating,
  author    = {Rajeev Alur and Thomas A. Henzinger and Orna Kupferman},
  title     = {Alternating-Time Temporal Logic},
  journal   = {Journal of the ACM},
  volume    = {49},
  number    = {5},
  pages     = {672--713},
  year      = {2002}
}

@article{abramsky2000full,
  author    = {Samson Abramsky and Radha Jagadeesan and Pasquale Malacaria},
  title     = {Full Abstraction for {PCF}},
  journal   = {Information and Computation},
  volume    = {163},
  number    = {2},
  pages     = {409--470},
  year      = {2000}
}

@techreport{ellison2007ceremony,
  author      = {Carl Ellison},
  title       = {Ceremony Design and Analysis},
  institution = {IACR},
  number      = {2007/399},
  year        = {2007}
}

@article{bolton2013formally,
  author    = {Matthew L. Bolton and Ellen J. Bass and Radu I. Siminiceanu},
  title     = {Using Formal Verification to Evaluate Human-Automation Interaction: A Review},
  journal   = {IEEE Transactions on Systems, Man, and Cybernetics: Systems},
  volume    = {43},
  number    = {3},
  pages     = {488--503},
  year      = {2013}
}

@inproceedings{esteva2001formal,
  author    = {Marc Esteva and Juan A. Rodr{\'i}guez-Aguilar and Carles Sierra and Pere Garcia and Josep Llu{\'i}s Arcos},
  title     = {On the Formal Specification of Electronic Institutions},
  booktitle = {Agent-Mediated Electronic Commerce (AMEC)},
  publisher = {Springer},
  series    = {LNCS},
  volume    = {1991},
  pages     = {126--147},
  year      = {2001}
}

@article{artikis2009specifying,
  author    = {Alexander Artikis and Marek Sergot and Jeremy Pitt},
  title     = {Specifying Norm-Governed Computational Societies},
  journal   = {ACM Transactions on Computational Logic},
  volume    = {10},
  number    = {1},
  pages     = {1--42},
  year      = {2009}
}

@article{shapiro2026bonds,
  title={Grassroots Bonds: A Grassroots Foundation for Market Liquidity}, 
  author={Shapiro, Ehud},
journal={arXiv preprint arXiv:2603.13671},
  year={2026}
}

@article{shapiro2026types,
  title={Types for Grassroots Logic Programs},
  author={Shapiro, Ehud},
  journal={arXiv preprint arXiv:2601.17957},
  year={2026}
}

@article{shapiro2026implementing,
  title={Implementing Grassroots Logic Programs with Multiagent Transition Systems and AI},
  author={Shapiro, Ehud},
  journal={arXiv preprint arXiv:2602.06934},
  year={2026}
}

@book{hoare1985communicating,
  author    = {Hoare, C. A. R.},
  title     = {Communicating Sequential Processes},
  publisher = {Prentice-Hall},
  year      = {1985},
  isbn      = {0-13-153271-5}
}

@online{torrentfreak2021bittorrent,
  title = {{BitTorrent Turns 20: The File-Sharing Revolution Revisited}},
  author = {{TorrentFreak}},
  year = {2021},
  month = {7},
  day = {2},
  url = {https://torrentfreak.com/bittorrent-turns-20-the-file-sharing-revolution-revisited-210702/},
  urldate = {2025-01-21},
  note = {Contains Bram Cohen's original statement: "BitTorrent's customer is etree. Etree is a loose-knit community of people who distribute live concert recordings online"}
}

@inproceedings{maymounkov2002kademlia,
  title = {Kademlia: A Peer-to-peer Information System Based on the XOR Metric},
  author = {Maymounkov, Petar and Mazi{\`e}res, David},
  booktitle = {Peer-to-Peer Systems: First International Workshop, IPTPS 2002},
  pages = {53--65},
  year = {2002},
  month = {March},
  publisher = {Springer},
  address = {Cambridge, MA, USA},
  doi = {10.1007/3-540-45748-8_5}
}

@article{shapiro2025glp,
  title={GLP: A Grassroots, Multiagent, Concurrent, Logic Programming Language},
  author={Shapiro, Ehud},
  journal={arXiv preprint arXiv:2510.15747},
  year={2025}
}

@article{shapiro2025GF,
  title={Grassroots Federation: Fair Governance of Large-Scale, Decentralized, Sovereign Digital Communities},
  author={Shapiro, Ehud and Talmon, Nimrod},
  journal={Proc. of AAMAS'26; arXiv preprint arXiv:2505.02208},
  year={2025}
}

@article{keidar2025constitutional,
  title={Constitutional Consensus},
  author={Keidar, Idit and Lewis-Pye, Andrew and Shapiro, Ehud},
  journal={arXiv preprint arXiv:2505.19216},
  year={2025}
}

@inproceedings{shapiro2025atomic,
  title={Grassroots Platforms with Atomic Transactions: Social Graphs, Cryptocurrencies, and Democratic Federations},
  author={Shapiro, Ehud},
  booktitle={Proceedings of the 27th International Conference on Distributed Computing and Networking},
  pages={71--81},
  doi = {10.1145/3772290.3772309},
  note = {arXiv preprint arXiv:2502.11299},
  year={2026}
}

@article{benet2014ipfs,
  title={Ipfs-content addressed, versioned, p2p file system},
  author={Benet, Juan},
  journal={arXiv preprint arXiv:1407.3561},
  year={2014}
}

@inproceedings{raman2019challenges,
  title={Challenges in the decentralised web: The mastodon case},
  author={Raman, Aravindh and Joglekar, Sagar and Cristofaro, Emiliano De and Sastry, Nishanth and Tyson, Gareth},
  booktitle={Proceedings of the internet measurement conference},
  pages={217--229},
  year={2019}
}

@inproceedings{buchegger2009peerson,
  title={PeerSoN: P2P social networking: early experiences and insights},
  author={Buchegger, Sonja and Schi{\"o}berg, Doris and Vu, Le-Hung and Datta, Anwitaman},
  booktitle={Proceedings of the Second ACM EuroSys Workshop on Social Network Systems},
  pages={46--52},
  year={2009}
}

@inproceedings{chockler2007constructing,
  title={Constructing scalable overlays for pub-sub with many topics},
  author={Chockler, Gregory and Melamed, Roie and Tock, Yoav and Vitenberg, Roman},
  booktitle={Proceedings of the twenty-sixth annual ACM symposium on Principles of distributed computing},
  pages={109--118},
  year={2007}
}

@inproceedings{lynch1988theory,
  title={A theory of atomic transactions},
  author={Lynch, Nancy and Merritt, Michael and Weihl, William and Fekete, Alan},
  booktitle={ICDT'88: 2nd International Conference on Database Theory Bruges, Belgium, August 31--September 2, 1988 Proceedings 2},
  pages={41--71},
  year={1988},
  organization={Springer}
}

@book{lynch1993atomic,
  title={Atomic transactions: in concurrent and distributed systems},
  author={Lynch, Nancy A and Merritt, Michael},
  year={1993},
  publisher={Morgan Kaufmann Publishers Inc.}
}

@inproceedings{lampson1981chapter,
  title={Chapter 11. atomic transactions},
  author={Lampson, Butler W},
  booktitle={Distributed Systems—Architecture and Implementation: an Advanced Course},
  pages={246--265},
  year={1981},
  organization={Springer}
}

@inproceedings{borgstrom2009compositional,
  title={A compositional theory for STM Haskell},
  author={Borgstr{\"o}m, Johannes and Bhargavan, Karthikeyan and Gordon, Andrew D},
  booktitle={Proceedings of the 2nd ACM SIGPLAN Symposium on Haskell},
  pages={69--80},
  year={2009}
}

@inproceedings{de2010liveness,
  title={Liveness of communicating transactions},
  author={De Vries, Edsko and Koutavas, Vasileios and Hennessy, Matthew},
  booktitle={Asian Symposium on Programming Languages and Systems},
  pages={392--407},
  year={2010},
  organization={Springer}
}

@inproceedings{de2010communicating,
  title={Communicating transactions},
  author={de Vries, Edsko and Koutavas, Vasileios and Hennessy, Matthew},
  booktitle={International Conference on Concurrency Theory},
  pages={569--583},
  year={2010},
  organization={Springer}
}

@inproceedings{acciai2007concurrent,
  title={A concurrent calculus with atomic transactions},
  author={Acciai, Lucia and Boreale, Michele and Dal Zilio, Silvano},
  booktitle={European Symposium on Programming},
  pages={48--63},
  year={2007},
  organization={Springer}
}

@article{chockler2021multi,
  title={Multi-shot distributed transaction commit},
  author={Chockler, Gregory and Gotsman, Alexey},
  journal={Distributed Computing},
  volume={34},
  pages={301--318},
  year={2021},
  publisher={Springer}
}

@inproceedings{bravo2019reconfigurable,
  title={Reconfigurable atomic transaction commit},
  author={Bravo, Manuel and Gotsman, Alexey},
  booktitle={Proceedings of the 2019 ACM Symposium on Principles of Distributed Computing},
  pages={399--408},
  year={2019}
}

@inproceedings{kermarrec2020gossiping,
  title={Gossiping with append-only logs in secure-Scuttlebutt},
  author={Kermarrec, Anne-Marie and Lavoie, Erick and Tschudin, Christian},
  booktitle={Proceedings of the 1st international workshop on distributed infrastructure for common good},
  pages={19--24},
  year={2020}
}

@inproceedings{tarr2019secure,
  title={Secure scuttlebutt: An identity-centric protocol for subjective and decentralized applications},
  author={Tarr, Dominic and Lavoie, Erick and Meyer, Aljoscha and Tschudin, Christian},
  booktitle={Proceedings of the 6th ACM conference on information-centric networking},
  pages={1--11},
  year={2019}
}

@article{halpern2024federated,
  title={Federated Assemblies},
  author={Halpern, Daniel and Procaccia, Ariel D and Shapiro, Ehud and Talmon, Nimrod},
  booktitle={AAAI '25},
  journal={Proc AAAI 2025; arXiv preprint arXiv:2405.19129},
  year={2024}
}

@article{lewis2023grassroots,
  title={Grassroots Flash: A Payment System for Grassroots Cryptocurrencies},
  author={Lewis-Pye, Andrew and Naor, Oded and Shapiro, Ehud},
  journal={arXiv preprint arXiv:2309.13191},
  year={2023}
}

@inproceedings{shapiro2023gsn, 
author={Shapiro, Ehud},
title={Grassroots Social Networking: Serverless, Permissionless Protocols for Twitter/LinkedIn/WhatsApp},
year = {2023},
isbn = {979-8-4007-0225-9/23/09},
publisher = {Association for Computing Machinery},
doi = {10.1145/3599696.3612898},
booktitle = {OASIS ’23},
location = {Rome, Italy},
}

@article{shapiro2024gc,
 title={Grassroots Currencies: Foundations for Grassroots Digital Economies},
  author={Shapiro, Ehud},
  journal={arXiv preprint arXiv:2202.05619},
  year={2024}
}

@inproceedings{shapiro2023grassrootsBA, 
author={Shapiro, Ehud},
 title={Grassroots Distributed Systems: Concept, Examples, Implementation and Applications (Brief Announcement)
},year = {2023},
publisher = {LIPICS},
booktitle = {37th International Symposium on Distributed Computing (DISC 2023). (Extended version: arXiv:2301.04391)},
address = {Italy},
notes={Extended version: arXiv preprint arXiv:2301.04391},
pages = {47:1, 47:7}
}

@article{shapiro2021multiagent,
  title={Multiagent Transition Systems: Protocol-Stack Mathematics for Distributed Computing},
  author={Shapiro, Ehud},
  journal={arXiv preprint arXiv:2112.13650},
  year={2021}
}

@inproceedings{shahaf2020genuine,
  title={Genuine Personal Identifiers and Mutual Sureties for Sybil-Resilient Community Growth},
  author={Shahaf, Gal and Shapiro, Ehud and Talmon, Nimrod},
  booktitle={International Conference on Social Informatics},
  pages={320--332},
  year={2020},
  organization={Springer},
  publisher={Springer},
  address={EU}
}

@inproceedings{chockler2007spidercast,
  title={Spidercast: a scalable interest-aware overlay for topic-based pub/sub communication},
  author={Chockler, Gregory and Melamed, Roie and Tock, Yoav and Vitenberg, Roman},
  booktitle={Proceedings of the 2007 inaugural international conference on Distributed event-based systems},
  pages={14--25},
  year={2007}
}

@inproceedings{douceur2002sybil,
  title={The sybil attack},
  publisher={Springer},
  address={EU},
  author={Douceur, John R},
  booktitle={Proceedings of the international workshop on peer-to-peer systems},
  pages={251--260},
  year={2002},
}

@misc{bitcoin,
  title={A peer-to-peer electronic cash system},
  author={Nakamoto, Satoshi and Bitcoin, A},
  journal={Bitcoin.--URL: https://bitcoin. org/bitcoin. pdf},
  volume={4},
  year={2008}
}

@article{levine2006survey,
  title={A survey of solutions to the sybil attack},
  author={Levine, Brian Neil and Shields, Clay and Margolin, N Boris},
  journal={University of Massachusetts Amherst},
  volume={7},
  pages={224},
  year={2006}
}

@inproceedings{alvisi2013sok,
  title={Sok: The evolution of sybil defense via social networks},
  author={Alvisi, Lorenzo and Clement, Allen and Epasto, Alessandro and Lattanzi, Silvio and Panconesi, Alessandro},
  booktitle={Proceedings of the 2013 {IEEE} symposium on security and privacy},
  pages={382--396},
  year={2013},
  organization={IEEE}
}

@misc{proof,
  title={Proof of Stake FAQ, Ethereum Wiki},
  author={Nakamoto, Satoshi and Bitcoin, A},
  howpublished = { \\ \MYhref{https://eth.wiki/en/concepts/proof-of-stake-faqs}{https://eth.wiki/en/concepts/proof-of-stake-faqs}},
  year={2019}
}

@misc{Cryptocurrencies,
  title={All Cryptocurrencies, CoinMarketCap},
  howpublished = { \MYhref{https://coinmarketcap.com/all/views/all/}{https://coinmarketcap.com/all/views/all/}},
  year={2020}
}

\mypara{Acknowledgements}
We thank Idit Keidar and Nimrod Talmon for our discussions and their feedback on this and related topics.

\appendix
\section{Systems Based on Shared Global Data Structures Are Not Grassroots}\label{appendix:not-grassroots}

\begin{proposition}\label{proposition:bitcoin}
Bitcoin is not grassroots.
\end{proposition}
\begin{proof}
We argue that Bitcoin is not oblivious.  The Bitcoin protocol specifies a set of bootnodes $B$, which know of and communicate with each other, each holding an initial chain consisting of the genesis block, and its liveness condition requires that blocks mined by correct agents are eventually propagated to all correct agents.  Consider disjoint $P, P' \subset \Pi$ such that $P$ and $P'$ have disjoint members of $B$, say $q\in P$ and $q'\in P'$.  In an interleaving of a correct run of~$P$ and a correct run of~$P'$, each step is either a $P$-step or a $P'$-step, so $q$ and $q'$ develop independent chains and no block propagation between them ever occurs.  Once $q$ mines a block, the propagation-to-$q'$ class becomes enabled and remains so indefinitely with no member ever taken; the interleaving is therefore not live, hence not a correct run of $P\cup P'$, and Bitcoin is not oblivious.
\end{proof}

\mypara{Distributed Hash Tables and IPFS} The same argument applies to systems based on distributed hash tables~\cite{maymounkov2002kademlia} and distributed file systems such as IPFS~\cite{benet2014ipfs}.  In a DHT, a lookup by a member of $P$ may require routing through a member of $P'$, so that a lookup transition that succeeds when $P$ runs alone does not correspond to a valid transition in the combined system where routing tables reflect both groups.  More precisely: in an interleaving of correct runs of $P$ and $P'$, each group builds its own routing table over its own members.  A lookup in the combined system $P\cup P'$ must route through the combined table, but in the interleaving no routing entry connecting the two groups is ever established, so lookups that should succeed in $P\cup P'$ (by routing through $P'$) instead fail or route incorrectly.  The liveness condition requiring that lookups for keys stored by correct agents eventually succeed is thus violated.  The argument for IPFS is analogous, as it relies on a DHT (specifically, a Kademlia-based DHT~\cite{maymounkov2002kademlia}) for content discovery.  Trackerless BitTorrent~\cite{torrentfreak2021bittorrent} relies on a Kademlia-based DHT for peer discovery, so the argument applies to it as well.

\section{Implementation}\label{section:implementation}

The specifications presented here provide the formal foundation for grassroots platforms. Here, we discuss how these platforms can be implemented, focusing on the role of GLP as the implementation language.

\mypara{GLP as implementation language}  Grassroots Logic Programs (GLP)~\cite{shapiro2025glp} is designed for implementing grassroots platforms on networked smartphones.  A correct multiagent implementation of GLP (madGLP)~\cite{shapiro2026implementing} has been developed and proven to correctly implement maGLP, with the additional result that correct and complete implementations preserve the grassroots property.  Thus, any grassroots platform that can be specified and proven grassroots at the transaction level, and correctly implemented in GLP, is guaranteed to remain grassroots at the implementation level.

\mypara{AI-derived implementations}  The mathematical foundations presented here and in companion papers have been used by AI to derive working implementations: (\ia)~a workstation-based implementation of concurrent GLP and a smartphone-based multiagent implementation of GLP, both in Dart, derived from the formal operational semantics~\cite{shapiro2026implementing}; (\ib)~a GLP implementation of the grassroots social graph, child-safe social networking with parental-consent-based befriending and group membership~\cite{shapiro2026volitional}; (\ic)~a GLP implementation of grassroots bonds, including a running six-agent village market scenario exercising symmetric and asymmetric credit, payments, redemption, escrow, and sale of debt~\cite{shapiro2026bonds}; and (\iiv)~a moded type system for GLP, implemented in Dart from a mathematical specification~\cite{shapiro2026types}.

\mypara{Enforcement}  Enforcement of the digital social contract---ensuring that participants cannot deviate from the protocol as programmed---is achieved via mutual attestation among the participants' machines, as described in companion work on secure GLP.

\mypara{Binary transactions} A standard way to realise binary transactions using unary transition systems is for one agent, say $p$, to \textsc{offer} the transaction to $q$, who may respond with \textsc{accept},
upon which $p$ may respond with \textsc{commit}, upon which the offered transaction is deemed to have been executed, or \textsc{abort}.
Agent $p$ may also issue \textsc{abort} before or after receiving any response from $q$ to its offer, provided $p$ has not previously issued \textsc{commit}. 

A challenge in this implementation is that a faulty $p$ may fail to either \textsc{commit} or \textsc{abort} following an \textsc{accept} by $q$, leaving $q$ in limbo, at least in regards to this transaction.  Solutions to this are a subject of future work.

For now, we note that, worst case, a friendship offer by $p$ accepted by $q$ would remain in limbo.  If it is committed by $p$ at some later point, which is not convenient to $q$, then $q$ can promptly unfriend $p$, with little or no harm done.
In the case of grassroots bonds, a swap transaction in limbo may tie bonds offered by $q$, which may or may not be harmful to $q$ (not harmful if these are $q$-bonds, which $q$ may mint as it pleases; or $p$-bonds that $q$ tries to redeem, and if $p$ is non-responsive it might indicate that $p$-bonds are not worth much anyhow).

\section{Formal Models of Persons in Concurrent Systems}\label{appendix:persons}

The formal methods tradition has a long lineage of modelling human agents as sources of nondeterminism alongside deterministic machines, but---to the best of our knowledge---without decomposing an agent into person and machine as components of its state.

\mypara{Turing's choice machines}
Before defining what we now call Turing machines (automatic machines, or a-machines), Turing~\cite{turing1936computable} introduced \emph{choice machines} (c-machines), ``whose motion is only partially determined by the configuration''---at designated states, the machine ``cannot go on until some arbitrary choice has been made by an external operator.''  The external operator is a person who freely chooses between alternatives; the sequence of choices determines which computation unfolds.  Turing immediately set c-machines aside, showing that any c-machine computation can be enumerated by an a-machine.  His 1939 oracle machines~\cite{turing1939systems} extend this further: the oracle ``cannot be a machine'' and provides answers the computation cannot derive internally.  Both formalisms model the person as \emph{external} to the machine, providing input at designated points.

\mypara{Process algebras}
Hoare's CSP~\cite{hoare1985communicating} provides the cleanest process-algebraic encoding of human choice.  External choice ($\Box$) offers the environment---potentially a person---a selection among initial events, while internal choice ($\sqcap$) is resolved by the system.  However, the person remains \emph{outside} the system boundary: CSP models what the person \emph{does}, not what the person \emph{is willing to do}.
Milner's CCS~\cite{milner1980calculus} uses a single summation operator without formally separating internal from external nondeterminism at the syntactic level.  In his later work on bigraphs~\cite{milner2009space}, Milner makes the scope explicit: agents ``can be artificial, as in computing systems\ldots\ or they can be natural, e.g.\ communicating humans.''  Both CSP and CCS model agents uniformly---there is no formal distinction between a person and a machine within the same agent.

\mypara{I/O automata and reactive systems}
Lynch and Tuttle's I/O automata~\cite{lynch1989introduction} partition actions into input (environment-controlled), output (automaton-controlled), and internal actions.  The key property is \emph{input-enabling}: an automaton cannot block input actions, so the environment---potentially a human operator---can act at any moment.  The Hybrid I/O Automata extension~\cite{lynch2003hybrid} explicitly states that HIOAs are ``intended to model all components of hybrid systems, including\ldots\ humans.''  The person, however, is part of the environment, not a component of the automaton's state.
Harel and Pnueli's reactive systems paradigm~\cite{harel1985development} draws the foundational dichotomy between transformational systems (batch, terminating) and reactive systems (ongoing interaction with environment).  The system is deterministic; all nondeterminism is attributed to the environment.  This paradigm was explicitly motivated by human-machine interaction, yet the formalism treats the human as environment rather than as a component of the system.

\mypara{Angelic and demonic nondeterminism}
The distinction between angelic and demonic nondeterminism, developed by Back and von Wright~\cite{back1998refinement} in the refinement calculus, provides a semantic treatment relevant to the person/machine boundary.  Demonic nondeterminism models adversarial environments (the worst-case choice is made); angelic nondeterminism models cooperative choices (the best-case choice is made).  This duality is the closest precursor to the volition/obligation distinction in the present work: a volitional transaction guarded by both parties (both must be willing) versus one guarded by either party (either can force it).  Two distinctions separate the frameworks.  First, the refinement calculus operates within a sequential program framework, not a multiagent transition system.  Second, angelic nondeterminism is a point-of-choice semantics: a choice is resolved locally at each transition, with no residue carried forward.  Volitions, in contrast, are persistent, inspectable state that accumulates across transitions; a guard condition reads an agent's record of willing over the history of the run, not a single local resolution.  This shift---from choice as a point-semantic primitive to choice as state---is what lets volitions be shared, compared, and reasoned about within the transition system, rather than external to it.

\mypara{Game structures and alternating-time temporal logic}
Module checking~\cite{kupferman2001module} models open systems with the environment fully adversarial.  Alternating-time temporal logic (ATL)~\cite{alur2002alternating} interprets formulas over concurrent game structures where multiple agents simultaneously choose actions.  Game semantics~\cite{abramsky2000full} models computation as dialogue between Proponent (program, following a deterministic strategy) and Opponent (environment, making free moves).  These frameworks treat agents as symmetric players but do not decompose a single agent into person and machine components.

\mypara{Ceremony analysis and human-interactive verification}
Ellison's ceremony analysis~\cite{ellison2007ceremony} extends security protocol analysis to include human participants as protocol nodes.  Bolton's Enhanced Operator Function Model (EOFM)~\cite{bolton2013formally} translates hierarchical human task models into state machines for model checking, with the human as the sole source of nondeterminism.  Both treat human nondeterminism as a source of error to be verified against, rather than as a source of legitimate volition to be formally recorded.

\mypara{Normative multiagent systems and electronic institutions}
Electronic institutions~\cite{esteva2001formal} model multiagent interaction as dialogical frameworks where human and software agents are treated uniformly as role-playing entities.  Normative multiagent systems~\cite{artikis2009specifying} use Event Calculus to specify societies where agents ``may fail to, or even choose not to, conform to the specifications.''  These approaches model norms that constrain agents, but do not decompose an agent's state into machine and volitional components, nor do they formalise the distinction between transactions requiring all parties to be willing and those that are obligatory once initiated.

\end{document}